\documentclass[Vancouver,Times2COL]{arxivnjd}

\usepackage{amsmath,amsfonts,amssymb,latexsym,mathtools,amsthm}
\usepackage{graphicx}
\usepackage{enumitem}
\usepackage{tikz}
\usetikzlibrary{calc,patterns,decorations.pathmorphing,decorations.markings,arrows.meta,positioning,shapes.geometric,backgrounds}
\usepackage[usenames,dvipsnames]{xcolor}
\usepackage{cite}
\usepackage{booktabs}
\usepackage{multirow}
\usepackage{makecell}
\usepackage{float}
\usepackage{algorithm}
\usepackage{algpseudocode}
\allowdisplaybreaks

\usepackage{style} 

\newcommand{\cmark}{$\checkmark$}
\newcommand{\xmark}{$\times$}
\newcommand{\pmark}{$\circ$}

\articletype{Research Article}%
\journal{Int. J. Robust Nonlinear Control}
\volume{00}
\copyyear{2026}
\startpage{1}
\makeatletter
\def\oddhead@titlepage@info{\vbox{%
		\hbox to \textwidth{%
			\fontsize{8}{10}\selectfont\color{navy}%
			Published in \textit{Int. J. Robust Nonlinear Control} \textbullet\ DOI: \href{https://doi.org/10.1002/rnc.70760}{10.1002/rnc.70760}%
			\hfill%
		}\par\vspace*{3\p@}%
		{\color{navy}\rule{\textwidth}{.5\p@}}\par\vspace*{9\p@}%
		\vspace*{20.5\p@}%
}}
\def\evenhead@titlepage@info{\oddhead@titlepage@info}
\makeatother

\begin{document}
	
	\title{Composite Adaptive Control Barrier Functions for Safety-Critical Systems with Parametric Uncertainty}
	
	\author{Mohammadreza Kamaldar}
	
	\authormark{Kamaldar, M.}
	\titlemark{Composite Adaptive Control Barrier Functions for Safety-Critical Systems}
	
	\address{\orgdiv{Department of Mechanical, Aerospace, and Biomedical Engineering}, \orgname{University of South Alabama}, \orgaddress{\state{AL}, \country{USA}}}
	
	\corres{Corresponding author: Mohammadreza Kamaldar \email{mkamaldar@southalabama.edu}}
	
	\abstract[Abstract]{Control barrier functions guarantee safety but require accurate system models; parametric uncertainty invalidates these guarantees. Existing robust methods maintain safety via worst-case bounds at the cost of performance, while modular learning schemes decouple estimation from safety and risk constraint violations during transients. This paper presents the composite adaptive control barrier function (CaCBF) algorithm for nonlinear control-affine systems with linear parametric uncertainty. The adaptation law is derived from a composite energy function integrating a logarithmic safety barrier, a control Lyapunov function, and a parameter-error term, creating a direct coupling between estimation accuracy and the safety margin. We prove three main results: (i) the safe set is forward invariant for all bounded parameters, without requiring persistence of excitation; (ii) the safety guarantee is robust to bounded errors in the state-derivative measurement; and (iii) all closed-loop signals are uniformly ultimately bounded. We further prove that the CaCBF admissible control set always contains the robust counterpart as a subset. Simulations of adaptive cruise control, an omnidirectional robot, and a planar drone traversing a narrow gate confirm that CaCBF recovers the performance margin surrendered by robust methods while maintaining strict safety throughout.}
	
	
	\jnlcitation{\cname{%
			\author{Kamaldar M.}}.
		\ctitle{Composite Adaptive Control Barrier Functions for Safety-Critical Systems with Parametric Uncertainty.} \cjournal{\it Int. J. Robust and Nonlinear Control} \cvol{2026}.}
	
	\maketitle
	
	\section{Introduction}
	\label{sec:intro}
	Control barrier functions (CBFs) provide a mathematical mechanism for enforcing safety in autonomous systems \cite{ames2019cbf, ames2016control}. By acting as a pointwise filter, a CBF monitors a nominal control signal and modulates it only when the system approaches the boundary of a safe set. This operation renders the safe set forward invariant if the underlying mathematical model accurately represents the physical system. However, physical systems inevitably deviate from their models due to parametric uncertainties, such as variable payloads, unmodeled friction, or aerodynamic drag. This discrepancy challenges the formal safety guarantees of standard CBFs. When the model parameters are incorrect, the calculated safe control input may drive the system into the unsafe region. This uncertainty is typically addressed through two primary paradigms: robust control, bounding the uncertainty, or adaptive control, estimating the parameters.
	
	Robust formulations, such as those established by Jankovic \cite{jankovic2018robust} and Zhao et al. \cite{zhao2020adaptive}, maintain safety by enforcing constraints against the worst-case realization of the uncertainty. Recent developments have extended robust CBFs to accommodate complex dynamics. For instance, Buch et al. \cite{buch2021robust} addressed sector-bounded inputs via second-order cone programs, Xiao and Belta \cite{xiao2021high} developed formulations for high-order relative degree systems, and Ames et al. \cite{ames2019cbf} introduced input-to-state safety conditions for bounded disturbances. While theoretically rigorous, robustness imposes a geometric cost. As the bounds on parameter uncertainty expand, the subset of control inputs certified as safe contracts. In scenarios with significant uncertainty, this conservatism can render the safe control set empty, causing the solver to fail even when a safe solution physically exists.
	Furthermore, robust methods compel the system to operate with reduced efficiency, braking earlier, or moving more slowly, regardless of whether the environment is hostile or benign \cite{zhou1998essentials}.
	
	To mitigate this conservatism, researchers utilize online estimation to reduce uncertainty bounds. Modular architectures, such as the robust estimator framework proposed by Das and Burdick \cite{das2025robust} or the measurement-robust CBFs developed by Dean et al. \cite{dean2020guaranteeing}, employ learning modules to refine parameter estimates over time. These methods separate the estimation process from the control logic. The estimator typically minimizes a prediction error metric, such as the mean squared error between predicted and observed states \cite{ljung1999system}. While this improves model accuracy, the separation of timescales allows for transient estimation errors independent of the safety boundary. A transient estimation error that occurs while the system is near the boundary can induce a safety violation before the estimator converges.
	
	Integrated approaches attempt to couple estimation and control more tightly. Taylor and Ames \cite{taylor2020adaptive} introduced Adaptive CBFs to preserve safety under parametric uncertainty, a formulation extended by Lopez and Slotine \cite{lopez2023unmatched} to handle unmatched uncertainties via certainty equivalence. Other variations include the dynamic penalty functions employed by Xiao et al. \cite{xiao2022adaptive}, the high-order adaptive formulations by Cohen et al. \cite{cohen2022high}, and the auxiliary systems used by Cheng et al. \cite{cheng2024safety} and Zhang et al. \cite{zhang2025safety} for specific applications like electric vehicles and marine vessels. Recently, Cheng et al. \cite{cheng2026resilient} combined adaptive CBFs with disturbance observers to manage unstructured noise. Despite these advances, many state-of-the-art methods adopt a modular structure where the adaptation law (e.g., recursive least squares or gradient descent) remains agnostic to the barrier function, as seen in \cite{gutierrez2024adaptive, lopez2020robust}. The learning process reduces parameter error globally rather than prioritizing accuracy in the critical regions where safety is threatened.
	
	Parallel to these efforts, data-driven and nonlinear control techniques seek to tighten robust bounds through experience. Zeng et al. \cite{zeng2024robust} applied composite learning for robotic manipulators, Gutierrez et al. \cite{gutierrez2025gp} utilized real-time Gaussian Process modeling, and Chriat et al. \cite{chriat2023optimality} leveraged reinforcement learning to optimize class-$\mathcal{K}$ functions. Similarly, Barrier Lyapunov Functions (BLFs), explored by Jiang et al. \cite{jiang2018adaptive} and others \cite{lafflitto2018barrier, rodrigues2021adaptive, an2017barrier}, enforce constraints on tracking errors. However, BLF methods generally require the control architecture to follow a strict backstepping structure \cite{krstic1995nonlinear}, which constrains their application to general, optimization-based safety filters for arbitrary nominal controllers.
	
	This paper introduces the \textit{composite adaptive control barrier function} (CaCBF) framework to resolve the conflict between robust conservatism and transient learning risks. Instead of appending an estimator to a fixed controller, we unify parameter estimation and safety certification into a single Lyapunov-based design. We construct a composite energy function that sums a logarithmic barrier potential, representing the proximity to constraint violation, a control Lyapunov function, and a quadratic parameter error term. We then derive an update law that strictly dissipates the total energy. This establishes a feedback loop where the adaptation is driven not only by prediction error but also by the gradient of the safety barrier. As the system approaches the boundary, the effective learning rate adjusts to prioritize the parameters required to maintain safety.
	
	This work makes the following contributions. We construct a composite energy function that unifies a logarithmic safety barrier, a control Lyapunov function, and a quadratic parameter-error term, and we derive the adaptation law directly from this function (Section~\ref{sec:adaptive_cbf}). Building on this construction, we prove that forward invariance of the safe set holds for all bounded parametric uncertainties without persistence of excitation, and for any $\kappa>0$ (Theorem~\ref{thm:main_results0} and Remark~\ref{rem:kappa_free}). We further establish that this safety guarantee is unconditionally robust to bounded errors in the state-derivative measurement, with no restriction on the noise magnitude (Theorem~\ref{thm:noise_robust}). Beyond safety, we prove that all closed-loop signals are uniformly ultimately bounded for any $\kappa$ satisfying the explicit computable condition~\eqref{eq:kappa_condition} (Theorem~\ref{thm:main_results} and Corollary~\ref{cor:kappa_explicit}). Finally, we prove that the CaCBF admissible control set always contains the robust one as a subset, which quantifies the reduction in conservatism relative to robust CBF (Theorem~\ref{prop:reduced_conservatism}).
	
	We validate the framework through three numerical case studies: adaptive cruise control with unknown drag, an omnidirectional robot avoiding obstacles under unknown drift, and a planar drone navigating a narrow gate under unknown crosswind. The results demonstrate that CaCBF achieves set utilization comparable to exact-model methods and outperforms robust baselines, confirming its ability to recover performance in uncertain environments without compromising safety. A comparison of the CaCBF architecture against representative prior adaptive CBF methods is given in Remark~\ref{rem:prior_comparison}.
	
	The remainder of the paper is organized as follows. Section~\ref{sec:problem_formulation} formalizes the problem. Section~\ref{sec:adaptive_cbf} derives the CaCBF framework and control law. Section~\ref{sec:analysis} proves the main theoretical results. Section~\ref{sec:simulation} presents numerical examples, and Section~\ref{sec:conclusion} discusses conclusions and future directions.
	
	\textbf{Notation.} We denote the set of real numbers by $\BBR$ and the $n$-dimensional Euclidean space by $\BBR^n$. The Euclidean norm is denoted by $\|\cdot\|$. For a symmetric matrix $P \in \BBR^{n \times n}$, we write $P \succ 0$ (resp. $P \succeq 0$) if $P$ is positive definite (resp. positive semidefinite). Given $P \succ 0$, the weighted Euclidean norm is $\|z\|_P \triangleq \sqrt{z^\rmT P z}$. The $n \times n$ identity matrix is denoted by $I_n$.
	
	We denote the limit from the left by $\lim_{\varepsilon \uparrow 0}$ and from the right by $\lim_{\varepsilon \downarrow 0}$. A continuous function $\alpha\colon [0,a) \to [0,\infty)$ is of \textit{class $\mathcal{K}$} if it is strictly increasing and $\alpha(0)=0$. If $a=+\infty$ and $\lim_{r\to+\infty}\alpha(r)=+\infty$, it is of \textit{class $\mathcal{K}_\infty$}. In addition, for $b>0$, a continuous function $\alpha\colon (-b,a) \to \BBR$ is of \textit{extended class $\mathcal{K}_\rme$} if $\alpha(0)=0$ and it is strictly increasing.
	
	Let $h\colon \BBR^n \to \BBR$, $f\colon \BBR^n \to \BBR^n$, and $G\colon \BBR^n \to \BBR^{n \times m}$ be continuously differentiable functions. 
	For $x\in\BBR^n$, let
	$\nabla h(x) \triangleq [ \frac{\partial h}{\partial x_{(1)}}~\cdots~\frac{\partial h}{\partial x_{(n)}} ]^\rmT\in\BBR^{n}$ denote the gradient vector of $h$ with respect to $x$.
	In addition, the Lie derivative of $h$ along $f$ is $L_f h(x) \triangleq \nabla h(x)^\rmT f(x)\in\BBR$. Similarly, the Lie derivative of $h$ along $G$ is  $L_G h(x) \triangleq \nabla h(x)^\rmT G(x)\in\BBR^{1\times m}$. Finally, for an optimization problem subject to $g(z) \leq 0$, a point $z_*$ is \textit{strictly feasible} if $g(z_*) < 0$.

	\section{Problem Formulation}
	\label{sec:problem_formulation}
	
	Consider the nonlinear control-affine system subject to linear parametric uncertainty given by
	\begin{equation}
		\dot{x} (t) = f(x(t)) + F(x(t))\theta_* + G(x(t))u(t), \label{eq:sys_dyn}
	\end{equation}
	where, for all $t\ge0$, $x(t) \in \mathcal{X} \subseteq \BBR^n$ is the state, and $u(t) \in \mathcal{U} \subseteq \BBR^m$ is the control input. The functions $f\colon \mathcal{X} \to \BBR^n$, $F\colon \mathcal{X} \to \BBR^{n \times p}$, and $G\colon \mathcal{X} \to \BBR^{n \times m}$ are known and continuously differentiable, while $\theta_* \in \BBR^p$ represents a vector of constant, unknown parameters.
	We impose the following assumption:
	
	\begin{enumerate}[label=(A\arabic*),left=10pt]
		\item \label{assum:theta} There exists a known constant $\theta_{\max} > 0$ such that  $\|\theta_*\| \leq \theta_{\max}$.
	\end{enumerate}
	
	Assumption \ref{assum:theta} implies the existence of a known compact convex set
	\begin{equation}
		\Theta \triangleq \{ \theta \in \BBR^p\colon \|\theta\| \leq \theta_{\max} \}\label{eq_Theta}
	\end{equation}
	such that $\theta_*\in\Theta$.
	
	\begin{remark}
		\label{rem:A1_justification}
		The requirement that $\theta_*$ be constant and bounded is standard in
		adaptive control \cite{sastry2011adaptive,ioannou2012robust} and is
		physically meaningful for a broad class of engineering systems.
		\textit{Constant parameters} arise naturally when the uncertain quantities
		are structural properties of the plant that are time-invariant over the
		operational horizon of interest: the aerodynamic drag and rolling-resistance
		coefficients in adaptive cruise control, the viscous friction or payload mass
		of a robotic manipulator, and the crosswind forces acting on a drone are all
		well-modeled this way, and we study each in Section~\ref{sec:simulation}.
		When true time-variation is present but slow, the constant-parameter
		assumption remains a valid approximation whose error can be bounded and
		treated as a residual disturbance; extending CaCBF to slowly time-varying
		parameters via a leakage modification is a direction for future work.
		\textit{Known bounds} are routinely available in practice: physical
		conservation laws (e.g., drag forces are bounded by dynamic pressure),
		manufacturer datasheets (e.g., maximum payload mass), or offline
		system-identification experiments all provide conservative envelopes
		$\theta_{\max}$ without requiring exact values. In fact, if the parameters
		were known exactly, no adaptation would be needed; the bound $\theta_{\max}$
		does not enter the control law but only initializes the projection set
		$\Theta$ and prevents parameter drift, thus a conservative overestimate is both
		safe and acceptable. The selection guideline for $\theta_{\max}$ is discussed
		further in Remark~\ref{rem:parameter_guidelines}.
	\end{remark}
	
	The dynamics in \eqref{eq:sys_dyn} thus capture a broad class of physical
	systems characterized by linear-in-parameter uncertainty.
	
	To encode the safety objective, let the \textit{safe set} $\mathcal{C}\subset\SX$ be the superlevel set of a continuously differentiable function $h\colon \SX \to \BBR$ defined by
	\begin{align}
		\mathcal{C} &\isdef \{ x \in \SX\colon h(x) \geq 0 \},
		\label{eq_safe_set}
	\end{align}
	and let $\bd \mathcal{C} \isdef \{ x \in \SX\colon h(x) = 0 \}$ and $\Int\mathcal{C} \isdef \{ x \in \SX\colon h(x) > 0 \}$ denote the boundary and interior of $\mathcal{C}$, respectively. 
	
	Let $s\colon[0,\infty)\times \SX\to\SX$ be such that, for all $t\ge0$, $s(t,x(0))$ is the  \textit{trajectory} of \eqref{eq:sys_dyn}  corresponding to each initial condition $x(0)\in\SX$. The primary objective is to synthesize a control law $u\colon\SX\to\SU$ such that  $\mathcal{C}$ is \textit{forward invariant} with respect to the dynamics \eqref{eq:sys_dyn}, that is, for all initial conditions $x(0) \in \mathcal{C}$ and all $t > 0$, $s(t,x(0)) \in \mathcal{C}$.
	Mathematically, a sufficient condition for forward invariance of $\SC$ with respect to \eqref{eq:sys_dyn}  is provided by Nagumo's Theorem \cite{ames2019cbf}, which requires the vector field $\dot x$ to point into the set at the boundary, that is, for all $x \in \bd \mathcal{C}$, $\nabla h(x)^\rmT\dot x\ge 0$. Equivalently, if, for all $x \in \bd \mathcal{C}$,
	\begin{equation}
		\dot{h}(x, u(x)) \geq 0,\label{cond_Nag}
	\end{equation}
	then, $\SC$ is forward invariant with respect to \eqref{eq:sys_dyn}.
	To ensure \eqref{cond_Nag} holds and to regularize the closed-loop behavior within $\Int\mathcal{C}$, CBFs enforce the stronger condition that, for all $x \in \mathcal{C},$
	\begin{equation}
		\dot{h}(x, u(x)) \geq -\alpha(h(x)), 
		\label{eq_hsafe}
	\end{equation}
	where $\alpha \in \mathcal{K}_\rme$ is locally Lipschitz.
	Note that \eqref{eq_hsafe} implies that as the state $x$ approaches the boundary of $\SC$ (i.e., as $h(x) \downarrow 0$), the permissible rate of decay of the safety margin approaches zero (i.e., $-\alpha(h(x)) \uparrow 0$). Based on the invariance condition \eqref{eq_hsafe}, we  define a CBF for the system \eqref{eq:sys_dyn} as follows.
	
	\begin{definition}[]
		\label{def:cbf}
		A continuously differentiable function $h\colon \SX \to \BBR$ is a CBF for the system \eqref{eq:sys_dyn} and the safe set $\mathcal{C}$ defined by \eqref{eq_safe_set},  if there exists a  locally Lipschitz function $\alpha \in \mathcal{K}_\rme$ such that, for all $x \in \mathcal{C}$, the set
		\begin{equation}
			\SU_{\rm cbf}(x) \isdef \left\{ \bar u \in \mathcal{U}\colon \dot{h}(x, \bar u) \geq -\alpha(h(x)) \right\}\label{eq_KCBF}
		\end{equation}
		is non-empty, where the time derivative $\dot{h}$ is along the trajectories of \eqref{eq:sys_dyn}, and is given by
		\begin{equation}
			\dot{h}(x, \bar u) = L_f h(x) + L_{F}h(x)\theta_* + L_G h(x)\bar u.
		\end{equation}
	\end{definition}
	
	The result connecting Definition \ref{def:cbf} to safety is summarized in Theorem \ref{thm:safety} in Section \ref{sec:analysis}.
	Note that \eqref{eq_KCBF} relies on $\dot{h}$, which explicitly depends on the parameter $\theta_*$. Since $\theta_*$ is unknown, the set $\SU_{\rm cbf}$ is not directly computable. This motivates the design of the composite adaptive CBF in the following section, which simultaneously estimates $\theta_*$, drives the state toward a desired equilibrium, and enforces safety.
	To characterize the control objective of stabilizing the system states to the origin, we consider the following definition.
	
	\begin{definition}[]
		\label{def:clf}
		A continuously differentiable function $V\colon \SX\to [0,\infty)$ is a CLF for the system \eqref{eq:sys_dyn} if $V$ is positive definite, radially unbounded, and there exists $\lambda>0$ such that, for all $x\in\SX,$
		\begin{equation}
			\inf_{u \in \mathcal{U}} \left( L_f V(x) + L_F V(x)\theta_*+ L_G V(x) u \right) \le -\lambda V(x).
			\label{eq:clf_condition}
		\end{equation}
	\end{definition}
	
	Note that condition \eqref{eq:clf_condition} ensures the existence of control inputs that drive the system state to the origin. To ensure the stabilization objective is compatible with the safety constraints, we consider the standard feasibility condition that the equilibrium point lies within the interior of the safe set, i.e., $h(0) > 0$. In the following section, we unify this stabilization objective with the safety guarantees of Definition \ref{def:cbf} to derive the proposed composite adaptive framework.
	
	\section{Composite Adaptive Control Barrier Functions}
	\label{sec:adaptive_cbf}
	
	In this section, we develop the composite adaptive control barrier function (CaCBF) framework. Unlike robust approaches that rely on static, worst-case bounds, we design an adaptation law that dynamically reduces conservatism while guaranteeing safety. We term this framework \textit{composite} because it unifies three coupled objectives: \textit{safety}, which enforces forward invariance via a logarithmic barrier; \textit{stability}, which drives the state to equilibrium; and \textit{adaptation}, driven by the minimization of an error-estimation cost, to allow the system to operate closer to the safety boundary without violating it.
	
	Let $\hat \theta\colon [0,\infty)\to\BBR^p$ be such that, for all $t\ge0$, $\hat{\theta}(t)$ denotes the instantaneous estimate of the unknown parameter $\theta_*$. In addition, let $\tilde \theta\colon [0,\infty)\to\BBR^p$, defined by $\tilde{\theta}(t) \isdef \hat{\theta}(t)-\theta_*$, denote the instantaneous  error.
	Using Lie derivatives and suppressing time dependence for brevity, the time derivative of $h$ along the trajectories of \eqref{eq:sys_dyn} is given by
	\begin{equation}
		\dot{h}(x) = L_f h(x) + L_{F}h(x)\theta_* + L_G h(x)u.
		\label{eq_doth}
	\end{equation}
	Substituting $\theta_*=\hat \theta-\tilde\theta$ into \eqref{eq_doth} yields
	\begin{align}
		\dot{h}(x) &= L_f h(x) + L_{F}h(x)(\hat{\theta} - \tilde{\theta}) + L_G h(x) u \nonumber \\
		&= \dot h_{\hat \theta}(x,\hat \theta) -\psi (x) \tilde{\theta},
		\label{eq:h_dynamics_split}
	\end{align}
	where $\psi\colon \SX \to \BBR^{1\times p}$ is the \textit{safety regressor} defined by
	\begin{equation}
		\psi(x) \isdef L_{F}h(x),
	\end{equation}
	and $\dot h_{\hat \theta}\colon\SX\times\BBR^p\to\BBR$ is defined by
	\begin{align}
		\dot h_{\hat \theta}(x,\hat\theta) \isdef L_f h(x) + \psi(x)\hat{\theta} + L_G h(x)u,\label{eq_doth_hat}
	\end{align}
	which represents the derivative of the barrier function based on the current parameter estimate $\hat \theta$.
	
	To simultaneously ensure safety, promote stability, and bound the parameter error, we introduce the CaCBF, denoted by $V_\rmc\colon \SX \times \BBR^p \to \BBR$ and defined as
	\begin{equation}
		V_\rmc(x, \hat{\theta}) \isdef -\ln\left(\frac{h(x)}{1+h(x)}\right) + \kappa V(x) + \frac{1}{2}\tilde{\theta}^\rmT \Gamma^{-1} \tilde{\theta},
		\label{eq:composite_barrier}
	\end{equation}
	where $h$ is the CBF used to define the safe set $\SC$ in \eqref{eq_safe_set}, $V\colon\SX\to[0,\infty)$ is a candidate CLF for the system \eqref{eq:sys_dyn}, $\kappa > 0$ is a weighting parameter, and $\Gamma \in \BBR^{p \times p}$ is a positive-definite adaptation gain matrix. Note that $V_\rmc$ contains a logarithmic barrier function. This construction relies on the property that the barrier term becomes unbounded as the system approaches the safety boundary, that is,
	\begin{equation*}
		\lim_{x\to\bd \SC} \ln\left(\frac{h(x)}{1+h(x)}\right) = -\infty.
	\end{equation*}
	Consequently, if $V_\rmc$ remains bounded, the state $x$ must remain within $\SC$---since any approach to $\bd\SC$ would force $V_\rmc\to+\infty$. This implication is the engine of the safety proof in Section~\ref{sec:analysis} (Theorems~\ref{thm:main_results0} and~\ref{thm:noise_robust}).
	
	We now derive the adaptation law $\dot{\hat{\theta}}$ by analyzing the time evolution of \eqref{eq:composite_barrier}.
	First, the time derivative of $V$ along the trajectories of \eqref{eq:sys_dyn} is given by
	\begin{align}
		\dot{V}(x) & = L_f V(x) + L_{F}V(x)\theta_* + L_G V(x)u \nonumber \\
		&= \dot V_{\hat \theta}(x,\hat{\theta}) -\phi (x) \tilde{\theta},
		\label{eq:h_dynamics_split2}
	\end{align}
	where $\phi\colon \SX \to \BBR^{1\times p}$ is the \textit{stability regressor} defined by
	\begin{equation}
		\phi(x) \isdef L_{F}V(x),
	\end{equation}
	and $\dot V_{\hat \theta}\colon\SX\times\BBR^p\to\BBR$ is defined by
	\begin{align}
		\dot V_{\hat \theta}(x,\hat{\theta}) \isdef L_f V(x) + \phi(x)\hat{\theta} + L_G V(x)u,\label{eq_dotV_hat}
	\end{align}
	which represents the time derivative of the CLF based on the current parameter estimate $\hat \theta$. 
	Differentiating the composite energy \eqref{eq:composite_barrier} along the trajectories of \eqref{eq:sys_dyn} yields
	\begin{equation}
		\dot{V}_\rmc(x, \hat{\theta}) = -\frac{\dot{h}(x)}{h(x)(1+h(x))} + \kappa \dot{V}(x) + \tilde{\theta}^\rmT \Gamma^{-1} \dot{\tilde{\theta}}.
		\label{eq_VCd}
	\end{equation}
	Since $\theta_*$ is constant, it follows that $\dot{\tilde{\theta}} = \dot{\hat{\theta}}$. Thus, substituting \eqref{eq:h_dynamics_split} and \eqref{eq:h_dynamics_split2} into \eqref{eq_VCd} yields
	\begin{align}
		\dot{V}_\rmc(x, \hat{\theta}) &= -\frac{\dot h_{\hat \theta}(x,\hat\theta) - \psi (x) \tilde{\theta}}{h(x)(1+h(x))} + \kappa \big(\dot V_{\hat \theta}(x,\hat{\theta}) - \phi (x) \tilde{\theta}\big) \nn\\
		&\quad + \tilde{\theta}^\rmT \Gamma^{-1} \dot{\hat{\theta}} \nonumber \\
		&= -\frac{\dot h_{\hat \theta}(x,\hat\theta)}{h(x)(1+h(x))} + \kappa \dot V_{\hat \theta}(x,\hat{\theta}) \nn\\
		&\quad - \tilde{\theta}^\rmT \Gamma^{-1} \left( \kappa \Gamma\phi(x)^\rmT - \frac{\Gamma\psi(x)^\rmT}{h(x)(1+h(x))} - \dot{\hat{\theta}} \right),
		\label{eq:Vc_grouped}
	\end{align}
	where the terms in the parenthesis capture the sign-indefinite effect of parameter uncertainty on both safety and stability. We design $\dot{\hat{\theta}}$ to not only cancel these terms but also minimize the estimation error.
	
	Define the instantaneous estimation error $e\colon \SX\times \BBR^p \to \BBR^n$ by
	\begin{equation}
		e(x, \hat{\theta}) \isdef \dot{x} - \left(f(x) + F(x) \hat{\theta} + G(x) u\right),
		\label{eq_err}
	\end{equation}
	which represents the discrepancy between the measured or estimated state derivative and the estimated state derivative using the current parameter estimate $\hat{\theta}$. Substituting \eqref{eq:sys_dyn} into \eqref{eq_err} yields
	\begin{equation}
		e(x, \hat{\theta}) = -F(x) \tilde{\theta}.
		\label{eq_err_sub}
	\end{equation}
	To minimize the magnitude of $e$, we incorporate a gradient descent term based on the instantaneous estimation-error cost $J\colon \SX\times \BBR^p\to [0,\infty)$, defined as $J(x, \hat{\theta}) = \frac{1}{2}\|e(x, \hat{\theta})\|^2$. Treating the measured state derivative $\dot{x}$ as independent of $\hat{\theta}$, the gradient of $J$ is derived from \eqref{eq_err_sub} as follows:
	\begin{align}
		\frac{\partial J(x, \hat{\theta})}{\partial \hat\theta} &= \left( \frac{\partial e(x, \hat{\theta})}{\partial \hat{\theta}} \right)^\rmT \frac{\partial J(x, \hat{\theta})}{\partial e(x, \hat{\theta})} \nonumber \\
		&= \left( \frac{\partial}{\partial \hat{\theta}} \left( -F(x) \tilde{\theta} \right) \right)^\rmT e(x, \hat{\theta}) \nonumber \\
		&= -F(x)^\rmT e(x, \hat{\theta}) .
		\label{DJJ}
	\end{align}
	To cancel the sign-indefinite bracket in \eqref{eq:Vc_grouped} and simultaneously reduce the estimation error, we design the adaptation law to explicitly cancel the safety and stability regressors while injecting the gradient update \eqref{DJJ}.
	Furthermore, to ensure Assumption \ref{assum:theta} is satisfied, we apply a projection operator. Specifically, the parameter update law is given by
	\begin{align}
		\dot{\hat{\theta}} (x, \hat{\theta})&= \SP_{\Theta} \Bigg( \Gamma \Bigg( \kappa \phi(x)^\rmT - \frac{\psi(x)^\rmT}{h(x)(1+h(x))} \nn\\
		&\qquad\quad + \gamma F(x)^\rmT e(x, \hat{\theta}) \Bigg),\,\hat{\theta} \Bigg),
		\label{eq:adaptation_law}
	\end{align}
	where $\hat\theta(0)\in\Theta$, $\gamma \ge 0$ is the adaptation gain, and $\SP_{\Theta}\colon \BBR^p\times\BBR^p\to\BBR^p$ is a smooth projection operator that confines $\hat{\theta}$ to the known set $\Theta$, and is defined as
	\begin{equation}
		\SP_{\Theta}(\tau, \hat{\theta}) \isdef
		\begin{cases}
			\tau, & \|\hat{\theta}\| < \theta_{\max} \text{ or } \hat{\theta}^\rmT \tau \le 0, \\
			\tau - \Gamma \dfrac{ \hat{\theta}^\rmT \tau}{\hat{\theta}^\rmT \Gamma \hat{\theta}}\hat{\theta}, & \text{otherwise},
		\end{cases}
		\label{eq:proj_def}
	\end{equation}
	where $\tau\in{\BBR^p}.$
	The projection operator $\SP_{\Theta}(\tau, \hat{\theta})$ modifies the update direction $\tau$ only when the estimate reaches the boundary of $\Theta$ (i.e., $\|\hat{\theta}\| = \theta_{\max}$) and the update direction $\tau$ points outside the set $\Theta$ (i.e., $\hat{\theta}^\rmT \tau >0$). To satisfy the Lyapunov dissipation property with a general positive-definite gain matrix $\Gamma \succ 0$, the projection is defined as an orthogonal projection onto the tangent plane of the boundary with respect to the $\Gamma^{-1}$-weighted Euclidean metric used in \eqref{eq:composite_barrier}.
	As shown in Lemma \ref{lem:proj_property} in the next section, for all $\tau\in{\BBR^p}$, $\SP_{\Theta}$ satisfies the projection property
	\begin{equation}
		\tilde{\theta}^\rmT \Gamma^{-1} (\SP_{\Theta}(\tau) - \tau) \le 0.\label{eq_proj_property0}
	\end{equation}
	Substituting \eqref{eq:adaptation_law} back into \eqref{eq:Vc_grouped} and employing \eqref{eq_err_sub} and \eqref{eq_proj_property0} implies that
	\begin{align}
		\dot{V}_\rmc(x, \hat{\theta}) &\leq -\frac{\dot{h}_{\hat \theta}(x, \hat{\theta})}{h(x)(1+h(x))} + \kappa \dot{V}_{\hat{\theta}}(x, \hat{\theta}) \nn\\
		&\quad - \gamma \|e(x, \hat{\theta})\|^2,
		\label{eq:Vc_final}
	\end{align}
	which highlights the core advantage of the formulation: the CaCBF evolution depends \textit{only} on known or estimated quantities, while the estimation error provides a nonpositive contribution.
	
	\begin{remark}
		\label{rem:xdot_filtering}
		The adaptation law \eqref{eq:adaptation_law} uses the state derivative $\dot x$ through the prediction error $e$, which is often not directly measured. Two complementary observations address this. First, when $\dot x$ is approximated by finite differences, $\dot x\approx(x_k-x_{k-1})/T_\rms$, the approximation error is bounded, and Theorem~\ref{thm:noise_robust} in the next section establishes that forward invariance is preserved for all bounded errors $\bar w$ with no smallness condition. Second, the dependence on $\dot x$ can be removed entirely by the filtered-regressor technique used in \cite{slotine1989composite}. Specifically, convolving the dynamics \eqref{eq:sys_dyn} with the impulse response of a stable, strictly proper filter and integrating by parts yields a filtered prediction error computable from $x$ and $u$ alone, without differentiation. Incorporating this filtered error into the composite energy  \eqref{eq:composite_barrier} while preserving the dissipation inequality \eqref{eq:dissipation_inequality} is a direction for future work.
		We further note that Proposition~\ref{prop:pe_convergence} is a continuous-time result. When $\dot x$ is obtained by finite differencing, the estimator is implemented on a sampled and delayed signal, and the rate $\sigma$ does not transfer directly from the continuous-time argument: the sampled implementation induces a delayed composition of contraction maps whose convergence rate depends on the sampling period $T_\rms$ and the delay, and degrades as either grows. Quantifying this dependence requires a sampled-data analysis; general treatments of sampled-data models appear in \cite{ackermann2012sampled,yuz2013sampled}, and related techniques for bounding convergence rates of delayed compositions of contraction maps have been developed in the consensus literature \cite{shi2025consensus}. Such an analysis is beyond the scope of the present continuous-time development. The safety guarantees of Theorems~\ref{thm:main_results0} and~\ref{thm:noise_robust} are unaffected, since they require only that the resulting error be bounded.
	\end{remark}
	Leveraging \eqref{eq:Vc_final}, we synthesize a controller that bounds $V_\rmc$ by constraining the estimated safety and stability derivatives. Specifically, imposing the sufficient conditions
	\begin{align}
		\dot{h}_{\hat{\theta}}(x,\hat{\theta}) &\geq -\alpha(h(x)), \nn\\
		\dot{V}_{\hat{\theta}}(x,\hat{\theta}) &\leq -\lambda V(x) + \delta,
		\label{eq:sufficient_conditions}
	\end{align}
	where $\lambda,\delta>0$ and $\alpha$ is a class $\mathcal{K}$ function, transforms \eqref{eq:Vc_final} into the dissipation inequality
	\begin{align}
		\dot{V}_\rmc(x, \hat{\theta}) &\leq \frac{\alpha(h(x))}{h(x)(1+h(x))} - \kappa \lambda V(x)  - \gamma \|e(x, \hat{\theta})\|^2 + \kappa \delta,
		\label{eq:dissipation_inequality}
	\end{align}
	where the upper bound ensures that $V_\rmc$ is bounded, thereby guaranteeing safety, as proven via contradiction in Section~\ref{sec:analysis} (Theorem~\ref{thm:main_results0}).
	
	To strictly enforce the conditions in \eqref{eq:sufficient_conditions} while minimizing control effort and determining the optimal relaxation $\delta$, we formulate the control synthesis as a convex optimization problem that integrates the parameter estimates directly into the constraints.
	For all $x \in \SX$ and all $\hat{\theta} \in \BBR^p$ evolving according to the adaptation law \eqref{eq:adaptation_law}, the control input $u_*\colon\SX\times\BBR^p$ is the unique solution to the CLF-CBF-QP given by
	\begin{align}
		\big({u}_*(x, \hat{\theta}),& \delta_*(x, \hat{\theta})\big) \nn\\
		&\isdef \operatorname*{argmin}_{(u,\delta) \in \BBR^m\times [0,\infty)} \ \tfrac{1}{2} u^\rmT R(x) u + \rho \delta^2, \label{eq:qp_cost_unified} \\
		\text{s.t.} \quad & L_f h(x) + \psi(x) \hat{\theta} + L_G h(x) u \geq -\alpha(h(x)), \label{eq:const_cbf} \\
		& L_f V(x) + \phi(x) \hat{\theta} + L_G V(x) u \leq -\lambda V(x) + \delta, \label{eq:const_clf}
	\end{align}
	where $R\colon\SX\to\BBR^{m\times m}$ is positive definite, $\delta\ge0$ is the relaxation variable for the stability constraint, $\rho>0$ is the stability-relaxation penalty, and $\lambda>0$ is the nominal convergence rate of the CLF.
	
	\begin{remark}
		\label{rem:parameter_guidelines}
		The CaCBF has seven design parameters, namely, $\Gamma$, $\gamma$, $\kappa$, $\lambda$, $\rho$, $\alpha$, and $\theta_{\max}$. The adaptation gain $\Gamma\succ0$ governs how fast $\hat\theta$ evolves; larger eigenvalues accelerate estimation but may excite high-frequency oscillations, so a diagonal $\Gamma$ with entries scaled to the expected magnitude of each $\theta_*$-component is recommended as a starting point. The gradient gain $\gamma\ge0$ scales the prediction-error term in \eqref{eq:adaptation_law}: setting $\gamma=0$ disables the identification signal, while increasing $\gamma$ improves estimation accuracy when persistency of excitation is present but does not degrade safety when it is absent. The barrier weight $\kappa>0$ does not affect forward invariance (See Theorems~\ref{thm:main_results0} and~\ref{thm:noise_robust}), but must satisfy $\kappa>c_\alpha/\varepsilon$ for uniform boundedness (See Theorem~\ref{thm:main_results}, condition~\eqref{eq:kappa_condition}); under Lemma~\ref{prop:asymptotic_dominance} with condition~\ref{eq_111}, the explicit sufficient bound $\kappa>\alpha_0\lambda/(2A^2)$ is given in Corollary~\ref{cor:kappa_explicit}, and in practice values $\kappa\in$ are effective when the CLF and CBF are normalized to similar scales. The CLF decay rate $\lambda>0$ is the desired exponential convergence rate for the nominal, known-parameter system; $\lambda$ should reflect the desired closed-loop speed and must satisfy $\lambda>2c_F\theta_{\max}$ if condition~\ref{eq_222} (i.e., linear-regressor growth) of Lemma~\ref{prop:asymptotic_dominance} is invoked. The stability penalty $\rho>0$ penalizes relaxation of the CLF constraint, with large $\rho$ prioritizing stability and small $\rho$ allowing larger relaxation when the CBF constraint is tight, which can be important near the safety boundary. The class-$\mathcal{K}$ function $\alpha$ determines how quickly the CBF condition is enforced near the boundary; the linear choice $\alpha(r)=\alpha_0 r$ with $\alpha_0>0$ is the most common, where smaller $\alpha_0$ yields a more conservative safety margin and larger $\alpha_0$ allows the system to operate closer to the boundary. Finally, the bound $\theta_{\max}$ should be any known conservative upper bound on $\|\theta_*\|$: an overestimate increases initial conservatism but does not affect the safety proof, whereas an underestimate that violates $\theta_*\in\Theta$ would invalidate the projection argument and must be avoided.
	\end{remark}
	
	\begin{algorithm}[t!]
		\caption{\enskip CaCBF Algorithm}\label{alg:cacbf}
		\begin{algorithmic}
			\State \textbf{Design Parameters:} Gains $\Gamma \succ 0, \gamma > 0, \kappa > 0$; Weights $R \succ 0, \rho > 0$; Decay Rate $\lambda > 0$
			\State \textbf{Design Functions:} Barrier $h$, Lyapunov $V$, Class $\mathcal{K}$ $\alpha$
			
			\State \textbf{Sensor Data:} State $x$, Response $\dot{x}$
			\State \textbf{Internal State:} Parameter Estimate $\hat{\theta}$
			\State \textbf{Control \& Adaptation:} $u_*$, Update $\dot{\hat{\theta}}$
			
			\State \textit{\textcolor{CadetBlue}{1. Regressors \& Estimation Error}}
			\State $\psi(x) \leftarrow L_F h(x)$, \ $\phi(x) \leftarrow L_F V(x)$
			\State $e(x, \hat{\theta}) \leftarrow \dot{x} - (f(x) + F(x)\hat{\theta} + G(x)u_*)$
			
			\State \textit{\textcolor{CadetBlue}{2. Evaluate Parameter Update}}
			\State $\tau \leftarrow \Gamma \left[ \kappa \phi(x)^\rmT - \frac{\psi(x)^\rmT}{h(x)(1+h(x))} + \gamma F(x)^\rmT e(x, \hat{\theta}) \right]$
			
			\State \textit{\textcolor{CadetBlue}{3. Apply Projection}}
			\State $\dot{\hat{\theta}} \leftarrow \mathcal{P}_{\Theta}(\tau, \hat{\theta})$
			
			\State \textit{\textcolor{CadetBlue}{4. Solve Adaptive CLF-CBF-QP}}
			\State $(u_*, \delta_*) \leftarrow \operatorname*{argmin}_{(u, \delta)} \frac{1}{2} u^\rmT R(x) u + \rho \delta^2$
			\State \quad \text{s.t. } $L_f h(x) + \psi(x) \hat{\theta} + L_G h(x) u \geq -\alpha(h(x))$
			\State \quad \quad \ \ $~~L_f V(x) + \phi(x) \hat{\theta} + L_G V(x) u \leq -\lambda V(x) + \delta$
			
			\State \textit{\textcolor{CadetBlue}{5. Update System}}
			\State Apply $u_*$ and integrate $\dot{\hat{\theta}}$
		\end{algorithmic}
	\end{algorithm}
	
	The closed-loop system formed by Algorithm~\ref{alg:cacbf} applied to \eqref{eq:sys_dyn} is analyzed in the following section, where we prove the safety, noise-robustness, uniform ultimate boundedness (UUB), and parameter-convergence properties stated in the introduction.
	
	\section{Closed-Loop Analysis}
	\label{sec:analysis}
	
	This section analyzes the closed-loop system formed by the plant dynamics \eqref{eq:sys_dyn}, the adaptation law \eqref{eq:adaptation_law}, and the adaptive controller \eqref{eq:qp_cost_unified}--\eqref{eq:const_clf}. We establish \emph{feasibility} of the optimization (Theorem~\ref{thm:feasibility}), \emph{forward invariance} of the safe set (Theorems~\ref{thm:main_results0} and~\ref{thm:noise_robust}), and \emph{uniform boundedness} of all signals (Theorem~\ref{thm:main_results}); under persistent excitation, the parameter error additionally converges to zero exponentially (Proposition~\ref{prop:pe_convergence}).
	
	To ensure the design is well-posed, we impose the following structural assumptions.
	
	\begin{enumerate}[label=(A\arabic*),left=10pt,start=2]
		\item \label{assum:C_nonempty_reldegree_1} The safe set $\mathcal{C}$ is non-empty, and $h$ has relative degree 1 with respect to \eqref{eq:sys_dyn}; that is, for all $x \in \mathcal{C}$, $L_G h(x) \neq 0$.
		\item \label{assum:gains}
		$R$ is continuous and uniformly positive definite; that is, there exist $\underline{r}, \overline{r} > 0$ such that, for all $x\in\SC$, $\underline{r} I_m \preceq R(x) \preceq \overline{r} I_m$.
		\item \label{assum:h_is_CBF_and_h(0)_greater_0} $h$ is a CBF for the system \eqref{eq:sys_dyn} and the safe set $\mathcal{C}$, and $h(0)>0$.
		\item \label{assum:clf}
		$V$ is a CLF for the system \eqref{eq:sys_dyn}.
	\end{enumerate}
	
	Assumption \ref{assum:C_nonempty_reldegree_1} grants the control input direct authority over the safety barrier. By requiring the relative degree to be 1, we guarantee that the actuator can instantaneously influence the time derivative $\dot{h}$ to steer the system away from the boundary. If this condition fails (i.e., if $L_G h(x) = 0$), the safety constraint becomes locally independent of $u$, potentially rendering the optimization infeasible. We address this specific challenge in the planar drone scenario of Example \ref{ex_Drone_Gate}, where the position-based constraints exhibit relative degree 2. In such cases, we employ backstepping techniques \cite{nguyen2016exponential,xiao2021high} to construct extended barrier functions that recover the required relative degree 1 property, thereby restoring direct control authority over the safety condition.
	Assumption \ref{assum:gains} ensures the optimization problem remains strictly convex. By bounding the eigenvalues of $R$ away from zero, we guarantee a unique optimal solution and prevent the control effort from becoming unbounded. Finally, Assumptions \ref{assum:h_is_CBF_and_h(0)_greater_0} and \ref{assum:clf} align the safety and stability objectives. Specifically, the condition $h(0) > 0$ implies the target equilibrium lies strictly within the safe set, ensuring the CLF drives the state to the origin without conflicting with the safety barrier.
	
	The following result, which follows from standard invariance arguments,  links the algebraic constraint in Definition \ref{def:cbf} to the physical safety of the system. See \cite{ames2019cbf} for a proof.
	
	\begin{theorem}[]
		\label{thm:safety}
		Consider the system \eqref{eq:sys_dyn} and the safe set $\mathcal{C}$.
		Assume that  \ref{assum:C_nonempty_reldegree_1} is satisfied, and let $h: \SX\to \BBR$ be a CBF for \eqref{eq:sys_dyn} and $\mathcal{C}$. In addition, assume that, for all $t \geq 0$, $u(t) \in \SU_{\rm cbf}(x(t))$. Then, $\mathcal{C}$ is forward invariant with respect to \eqref{eq:sys_dyn}.
	\end{theorem}
	
	\begin{remark}
		\label{rem:thm1_role}
		Theorem~\ref{thm:safety} is a classical result that requires that, for all $t\ge0$, $u(t)\in\SU_{\rm cbf}(x(t)),$  where $\SU_{\rm cbf}$ is defined using the \emph{unknown} $\theta_*$. Because $\theta_*$ is not available, the adaptive controller cannot directly enforce membership in $\SU_{\rm cbf}$. The CaCBF framework addresses this by providing an \emph{independent} safety proof (i.e., Theorem~\ref{thm:main_results0}) that does not invoke Theorem~\ref{thm:safety} at all. Instead, forward invariance follows from the fact that the composite Lyapunov function $V_{\rm c}$ remains bounded, which in turn forces $h(x(t))>0$ for all $t\ge0$. Theorem~\ref{thm:safety} is included for completeness and to motivate the CBF framework, while the actual safety guarantee of the adaptive scheme is established via Theorem~\ref{thm:main_results0}.
	\end{remark}
	
	Having established the conditions for safety via Theorems~\ref{thm:main_results0} and~\ref{thm:noise_robust}, the next theorem confirms that the control law remains feasible and Lipschitz continuous.
	
	\begin{theorem}[]
		\label{thm:feasibility}
		Consider the CLF-CBF-QP defined in \eqref{eq:qp_cost_unified}--\eqref{eq:const_clf}. Assume that \ref{assum:C_nonempty_reldegree_1} and \ref{assum:gains} are satisfied. Then, the following statements hold:
		\begin{enumerate}
			\item\label{Thm2_p1} For all $x \in \Int\SC$ and all $\hat{\theta} \in \BBR^p$, the CLF-CBF-QP is strictly feasible and has a unique global minimizer.
			\item\label{Thm2_p2} The optimal control $u_*$ and  relaxation  $\delta_*$ are locally Lipschitz continuous on $\Int\SC \times \BBR^p$.
		\end{enumerate}
	\end{theorem}
	
	\begin{proof}
		Let $z \isdef [u^\rmT~ \delta]^\rmT \in \BBR^{m+1}$. We write the CLF-CBF-QP as
		\begin{align}
			z_*(x, \hat{\theta}) \isdef \operatorname*{argmin}_{z \in \BBR^{m+1}} \quad & \frac{1}{2} z^\rmT \mathcal{Q}(x) z, \label{eq:qp_standard} \\
			\text{s.t.} \quad & A(x) z \leq b(x, \hat{\theta}), 
		\end{align}
		where 
		\begin{gather*}
			\mathcal{Q}(x) \isdef \begin{bmatrix} R(x) & 0 \\ 0 & \rho \end{bmatrix}\in\BBR^{(m+1)\times(m+1)}, \\
			A(x) \isdef \begin{bmatrix} a_1(x) \\ a_2(x) \end{bmatrix}\in\BBR^{2\times (m+1)}, \quad
			b(x, \hat{\theta}) \isdef \begin{bmatrix} b_1(x, \hat{\theta}) \\ b_2(x, \hat{\theta}) \end{bmatrix}\in\BBR^{2},
		\end{gather*}
		and 
		\begin{align}
			a_1(x) &\isdef \matl-L_G h(x) & 0\matr\in\BBR^{1\times(m+1)}, \nn\\
			a_2(x) &\isdef \matl L_G V(x) & -1\matr\in\BBR^{1\times(m+1)}, \nn\\
			b_1(x, \hat{\theta}) &\isdef L_f h(x) + L_{F} h(x) \hat{\theta} + \alpha(h(x))\in\BBR, \nn\\
			b_2(x, \hat{\theta}) &\isdef  -L_f V(x) - L_{F} V(x) \hat{\theta} - \lambda V(x)\in\BBR.
			\label{eq_b2}
		\end{align}
		
		To prove \ref{Thm2_p1}, let $x \in \Int\SC$ and $\hat{\theta} \in \BBR^p$. 
		First, consider the safety constraint $a_1(x)z < b_1(x,\hat\theta)$. Since \ref{assum:C_nonempty_reldegree_1} is satisfied, it follows that $a_1(x) \ne0$, and thus the strict inequality $a_1(x) z < b_1(x,\hat\theta)$ defines a non-empty, open half-space in $\BBR^m$.
		Therefore, there exists $u_{*}\in\BBR^m$ such that, for all $\delta_0\ge0$, $z_{0}\isdef[ u_{*}^\rmT ~ \delta_0 ]^\rmT$ satisfies the strict safety constraint $a_1(x) z_{0} < b_1(x,\hat\theta)$.
		Now, let $\delta_1 \isdef L_G V(x) u_{*} - b_2(x,\hat\theta)$ and  $z_2\isdef[ u_{*}^\rmT~\delta_2 ]^\rmT $, where $\delta_2>\max(0,\delta_1)$. It thus follows that $a_2(x) z_2 < b_2(x,\hat\theta)$. Therefore, since $z_{2}$ satisfies both constraints, the CLF-CBF-QP problem is strictly feasible.
		Moreover, since $\rho > 0$ and \ref{assum:gains} is satisfied, it follows that $\mathcal{Q}(x) \succ 0$, ensuring a unique global minimizer $z_*$.
		
		To prove \ref{Thm2_p2}, we use the sensitivity theorem in \cite[Theorem 1]{morris2013sufficient}, which requires the linear independence constraint qualification. Note that $\mathcal{Q}$, $A$, and $b$ are composed of smooth vector fields $f, F, G$ and continuously differentiable functions $h$ and $V$. Since continuously differentiable functions are locally Lipschitz on compact sets, it follows that $\mathcal{Q}$, $A$, and $b$ are locally Lipschitz on compact sets. In addition, $R \succ 0$ and $\rho > 0$ imply strict convexity.
		Next, let $x\in\SC$ and $\hat\theta\in\BBR^p$. We prove that the gradients of the active constraints are linearly independent. First, consider the case where exactly one constraint is active. Since \ref{assum:C_nonempty_reldegree_1} is satisfied, it follows that $a_1(x) \neq 0$. Since, in addition, $a_2(x) \neq 0$, it follows that the active constraint is linearly independent.  
		Second, consider the case where both constraints are active. Let $c_1,c_2\in\BBR$, and consider the linear combination $c_1 a_1(x)^\rmT + c_2 a_2(x)^\rmT = 0$, which implies
		\begin{equation}
			c_1 \begin{bmatrix} -L_G h(x)^\rmT \\ 0 \end{bmatrix} + c_2 \begin{bmatrix} L_G V(x)^\rmT \\ -1 \end{bmatrix} = \begin{bmatrix} 0 \\ 0 \end{bmatrix}.\label{eq_const4}
		\end{equation}
		Since \ref{assum:C_nonempty_reldegree_1} is satisfied, \eqref{eq_const4} implies $c_1=c_2=0$.
		Thus, the gradients of the active constraints are linearly independent. Since all conditions of \cite[Theorem 1]{morris2013sufficient} are satisfied, $z_*$ is locally Lipschitz continuous, which confirms \ref{Thm2_p2} because $u_*$ and $\delta_*$ are linear projections of $z_*$.
	\end{proof}
	
	We now confirm that the projection operator \eqref{eq:proj_def} enforces parameter boundedness without corrupting the adaptation direction. This property ensures that the projection mechanism does not counteract the adaptation process. The following result is an extension of the standard projection-operator property from \cite{ioannou2012robust,lavretsky2012robust} to the weighted inner product induced by $\Gamma^{-1}$.
	
	\begin{lemma}[]
		\label{lem:proj_property}
		Assume that \ref{assum:theta} is satisfied. Then, for all $\tau \in \BBR^p$ and all $\hat{\theta} \in \Theta$, the projection operator $\SP_{\Theta}$ defined in \eqref{eq:proj_def} satisfies
		\begin{equation}
			\tilde{\theta}^\rmT \Gamma^{-1} \big(\SP_{\Theta}(\tau, \hat{\theta}) - \tau\big) \le 0.
			\label{eq_proj_property}
		\end{equation}
	\end{lemma}
	
	The following lemma guarantees that the projection operator  confines the parameter estimates to the compact set $\Theta$. Its proof follows the standard argument in \cite{ioannou2012robust,lavretsky2012robust} adapted to the $\Gamma^{-1}$-weighted metric.
	
	\begin{lemma}[]
		\label{lem:param_boundedness}
		Assume that \ref{assum:theta} is satisfied.
		Consider the parameter update law given by \eqref{eq:adaptation_law}, where $\hat{\theta}(0) \in \Theta$. Then, for all $t \ge 0$, $\hat{\theta}(t) \in \Theta$. 
	\end{lemma}
	
	With the parameter estimates strictly bounded, we turn to the safety of the system. The following theorem proves that the adaptive controller renders the safe set $\SC$ forward invariant with respect to the system \eqref{eq:sys_dyn}, ensuring the constraints are satisfied for all time.
	
	\begin{theorem}[]
		\label{thm:main_results0}
		Consider the closed-loop system comprising the system \eqref{eq:sys_dyn}, the parameter update law \eqref{eq:adaptation_law}, and  the adaptive control law $u_*$ defined by \eqref{eq:qp_cost_unified}--\eqref{eq:const_clf}. Assume that \ref{assum:theta}--\ref{assum:clf} are satisfied, and let $x(0) \in \Int\mathcal{C}$ and $\hat{\theta}(0) \in \Theta$. Then, the safe set $\mathcal{C}$ is forward invariant with respect to \eqref{eq:sys_dyn}.
	\end{theorem}
	
	\begin{proof}
		Consider the composite function $V_\rmc$ defined in \eqref{eq:composite_barrier}. Taking the time derivative along the closed-loop trajectories yields \eqref{eq:Vc_grouped}. 
		Note that since \ref{assum:theta} is satisfied and $\hat{\theta}(0)\in\Theta$, Lemma \ref{lem:param_boundedness} implies that for all $t\ge0$, $\hat{\theta}(t)\in{\Theta}$.
		Therefore, substituting the adaptation law \eqref{eq:adaptation_law} into \eqref{eq:Vc_grouped}, and using Lemma \ref{lem:proj_property} and \eqref{eq_err_sub} implies that
		\begin{align}
			\dot{V}_\rmc(x, \hat{\theta}) 
			&=-\frac{\dot h_{\hat \theta}(x,\hat\theta)}{h(x)(1+h(x))} + \kappa \dot V_{\hat \theta}(x,\hat{\theta})+\gamma \tilde{\theta}^\rmT F(x)^\rmT e(x,\hat\theta)\nn\\
			&\quad- \tilde{\theta}^\rmT \Gamma^{-1} \Bigg( \kappa \Gamma\phi(x)^\rmT - \frac{\Gamma\psi(x)^\rmT}{h(x)(1+h(x))} \nn\\
			&\qquad\qquad\quad +\gamma\Gamma F(x)^\rmT e(x,\hat\theta)- \dot{\hat{\theta}} \Bigg)\nn\\
			&\le -\frac{\dot h_{\hat \theta}(x,\hat\theta)}{h(x)(1+h(x))} + \kappa \dot V_{\hat \theta}(x,\hat{\theta})+\gamma \tilde{\theta}^\rmT F(x)^\rmT e(x,\hat\theta)\nn\\
			&= -\frac{\dot h_{\hat \theta}(x,\hat\theta)}{h(x)(1+h(x))} + \kappa \dot V_{\hat \theta}(x,\hat{\theta})-\gamma \| e(x,\hat\theta)\|^2,
			\label{eq:Vc_grouped2}
		\end{align}
		where the inequality follows from Lemma~\ref{lem:proj_property}, that is, \eqref{eq_proj_property} implies $\tilde\theta^\rmT\Gamma^{-1}(\SP_\Theta(\tau,\hat\theta)-\tau)\le0$, making the subtracted bracket in the second line of the derivation nonnegative; and the last step uses $\tilde\theta^\rmT F^\rmT e = \tilde\theta^\rmT F^\rmT(-F\tilde\theta)=-\|e\|^2$ from \eqref{eq_err_sub}.
		Note that \ref{assum:h_is_CBF_and_h(0)_greater_0} and \ref{assum:clf} are satisfied. Thus, using \eqref{eq_doth_hat} and \eqref{eq_dotV_hat}, and substituting \eqref{eq:const_cbf} and \eqref{eq:const_clf} into \eqref{eq:Vc_grouped2} yields
		\begin{align}
			\dot{V}_\rmc(x, \hat{\theta}) 
			&\le \frac{\alpha( h(x))}{h(x)(1+h(x))} - \kappa \lambda V(x)  -\gamma \| e(x,\hat\theta)\|^2+\kappa\delta_*(x, \hat{\theta}).
			\label{eq:Vc_grouped3}
		\end{align}
		
		Next, suppose for contradiction, that $\mathcal{C}$ is not forward invariant with respect to \eqref{eq:sys_dyn}. Then, there exists $t_1 > 0$ such that $h(x(t_1)) = 0$ and, for all $t \in [0, t_1)$, $h(x(t)) > 0$. 
		It thus follows from \eqref{eq:composite_barrier} that
		\begin{equation}
			\lim_{t \uparrow t_1} V_\rmc(x(t), \hat{\theta}(t)) = +\infty.
			\label{eq:Vc_blowup}
		\end{equation}
		Next, note that \eqref{eq:Vc_grouped3} implies that, for all $t\ge0,$
		\begin{equation}
			\dot{V}_\rmc(x(t),\hat\theta(t)) \le \frac{\alpha(h(x(t)))}{h(x(t))(1+h(x(t)))} + \kappa \delta_*(x(t),\hat\theta(t)).
			\label{eq:Vc_simplified}
		\end{equation}
		Since $\alpha$ is Lipschitz with $\alpha(0)=0$ and $h(x(t_1))=0$, there exists $c_\alpha>0$ such that $\alpha(h(x(t)))\le c_\alpha h(x(t))$ for all $t\in[0,t_1]$. Combined with $h(x(t))\ge0$ on this interval, \eqref{eq:Vc_simplified} gives, for all $t\in[0,t_1]$, 
		\begin{align}
			\dot{V}_\rmc(x(t),\hat\theta(t)) &\le c_\alpha + \kappa \delta_*(x(t),\hat\theta(t)).
			\label{eq:Vc_simplified2}
		\end{align}
		Next, note that since $x$ is continuous on the closed interval $[0, t_1]$, it follows that, for all $t \in [0, t_1],$ $x(t)$ belongs to a compact set. Moreover, since \ref{assum:theta} is satisfied, Lemma \ref{lem:param_boundedness} implies that $\hat{\theta}$ is bounded. Furthermore, since \ref{assum:C_nonempty_reldegree_1} and \ref{assum:gains} are satisfied,  part \ref{Thm2_p2} of Theorem \ref{thm:feasibility} implies $\delta_*$ is Lipschitz continuous. It thus follows that $\delta_{\max} \isdef \sup_{\tau \in [0, t_1]} \delta_*(x(\tau), \hat{\theta}(\tau))$ exists. 
		It thus follows from \eqref{eq:Vc_simplified2} that, for all $t\in[0,t_1]$,  
		\begin{align}
			\dot{V}_\rmc(x(t),\hat\theta(t)) 
			& \le c_\alpha + \kappa \delta_{\max}.
			\label{eq:Vc_simplified4}
		\end{align}
		Integrating \eqref{eq:Vc_simplified4} over the interval $[0, t_1]$ yields
		\begin{equation}
			V_\rmc(x(t_1),\hat\theta(t_1)) \le V_\rmc(x(0),\hat\theta(0)) + t_1(c_\alpha + \kappa \delta_{\max}) < +\infty,
		\end{equation}
		which contradicts \eqref{eq:Vc_blowup}. Thus, for all $t\ge0$, $h(x(t))>0$, which confirms the result.
	\end{proof}
	
	\begin{remark}
		\label{rem:kappa_free}
		Theorem~\ref{thm:main_results0} places no condition on $\kappa$: forward invariance holds for any $\kappa>0$. This is because the proof uses a finite-time contradiction argument in which $\kappa$ appears only in the bound $c_\alpha+\kappa\delta_{\max}$, which is finite for any fixed $\kappa$ since $\delta_*$ is Lipschitz continuous (Theorem~\ref{thm:feasibility}). The role of $\kappa$ is therefore purely to govern the tightness of the UUB set in Theorem~\ref{thm:main_results} (via condition~\eqref{eq:kappa_condition}), and to weight the relative priority of the CLF term in $V_{\rm c}$~\eqref{eq:composite_barrier}. A larger $\kappa$ drives the state more aggressively toward the origin outside $\mathcal{S}$ but does not add or remove the safety guarantee.
	\end{remark}
	
	To guarantee the uniform boundedness of all closed-loop signals, we impose the following structural assumption on the system's asymptotic behavior.
	
	\begin{enumerate}[label=(A\arabic*),left=10pt,start=6]
		\item \label{assum:feasibility}
		There exists a compact set $\mathcal{S} \subset \mathcal{X}$ containing the origin and a constant $\varepsilon > 0$ such that, for all $x \in \mathcal{X} \setminus \mathcal{S}$ and all $\hat{\theta} \in \Theta$,
		\begin{equation}
			\lambda V(x) - \delta_*(x, \hat{\theta}) \ge \varepsilon.
		\end{equation}
	\end{enumerate}
	
	Assumption \ref{assum:feasibility} ensures that outside the set $\mathcal{S}$, the drive for stability $\lambda V$ outpaces the cost of safety $\delta_*$. Physically, the relaxation term $\delta_*$ represents the control effort ``wasted'' to satisfy safety constraints or compensate for parameter errors. By requiring the stability term to dominate, we ensure the controller always retains enough authority to pull the system back toward the origin. This structure is typical of mechanical systems where the CLF grows quadratically, while the uncertainty and safety conflicts grow only linearly or sub-quadratically.
	
	The following result identifies sufficient conditions for this property to hold.
	
	\begin{lemma} \label{prop:asymptotic_dominance}
		Assume that  \ref{assum:theta}--\ref{assum:gains} and \ref{assum:clf} are satisfied, and let $\lambda>0$ and $u_\rmn\colon \SX \to \BBR^m$ be such that, for all $x\in\SX$, 
		\begin{equation}
			L_f V(x) + \phi(x)\theta_*+ L_G V(x) u_{\rm n}(x) \le -\lambda V(x).\label{eq_prop1V}
		\end{equation}
		In addition, assume that there exists $c_{u} > 0$ such that, for all $x \in \SX$, 
		\begin{equation}
			\|u_{\rm n}(x)\|_{R(x)} \le c_{u} \sqrt{V(x)}.\label{eq_unom2}
		\end{equation} 
		Furthermore, assume that there exists $c_{F} > 0$ such that, for all $x \in \SX$, at least one of the following conditions holds:
		\begin{enumerate}
			\item\label{eq_111}   $\|\phi(x)\| \le c_{F} \sqrt{V(x)}$. 
			\item\label{eq_222}  $\|\phi(x)\| \le c_{F} V(x)$ and $\lambda > 2c_{F} \theta_{\max}$.
		\end{enumerate}
		Then,  \ref{assum:feasibility} is satisfied.
	\end{lemma}
	
	\begin{proof}
		Note that since \ref{assum:theta} is satisfied, it follows that $\Theta$ defined by \eqref{eq_Theta} exists. 
		Let $\hat{\theta}\in{\Theta}.$
		Since \ref{assum:clf} is satisfied, using \eqref{eq:h_dynamics_split2} and \eqref{eq_prop1V} implies that, for all $x\in\SX$,
		\begin{align}
			\dot{V}_{\hat{\theta}}(x, \hat{\theta}) &= \dot{V}(x) + \phi(x)\tilde{\theta} \nn \\
			&\le -\lambda V(x) + \phi(x)\tilde{\theta} \nn \\
			&\le -\lambda V(x) + \delta_{\rm n}(x, \hat{\theta}),
		\end{align}
		where $\delta_{\rm n}(x, \hat{\theta}) \isdef \max\{0, \phi(x)\tilde{\theta}\}$. 
		It thus follows that, for all $x\in\SX$, the pair $(u_{\rm n}(x), \delta_{\rm n}(x,\hat{\theta}))$ satisfies the stability constraint \eqref{eq:const_clf}. Moreover, since \ref{assum:C_nonempty_reldegree_1} and \ref{assum:gains} are satisfied, it follows that,  for all $x\in\SX$, $(u_*(x,\hat{\theta}), \delta_*(x,\hat{\theta}))$ is the optimal solution minimizing the cost \eqref{eq:qp_cost_unified}. It thus follows that, for all $x\in\SX$, 
		\begin{align}
			\rho \delta_*(x, \hat{\theta})^2 &\le \frac{1}{2} \|u_*(x, \hat{\theta})\|_{R(x)}^2 + \rho \delta_*(x, \hat{\theta})^2 \nn\\
			&\le \frac{1}{2} \|u_{\rm n}(x)\|_{R(x)}^2 + \rho \delta_{\rm n}(x, \hat{\theta})^2.\label{eq_rhodelta}
		\end{align}
		Note that, since $\hat{\theta}, \theta_* \in \Theta$, using the triangle inequality $\|\tilde{\theta}\| \le \|\hat{\theta}\| + \|\theta_*\| \le 2\theta_{\max}$ implies that, for all $x\in\SX$,  
		\begin{equation}
			\delta_{\rm n}(x, \hat{\theta}) \le \|\phi(x)\| \|\tilde{\theta}\| \le 2\|\phi(x)\| \theta_{\max}.\label{eq_delt_nom1}
		\end{equation}
		Taking the square root of both sides of \eqref{eq_rhodelta}, and  dividing by $\sqrt{\rho}$ implies that, for all $x\in\SX$,
		\begin{align}
			\delta_*(x, \hat{\theta}) & \le \frac{\|u_{\rm n}(x)\|_{R(x)}}{\sqrt{2\rho}}   +  \delta_{\rm n}(x, \hat{\theta})\nn\\
			&  \le \frac{c_u \sqrt{V(x)}}{\sqrt{2\rho}}  + 2\theta_{\max} \|\phi(x)\|,\label{eq_dleta3}
		\end{align}
		where the last inequality follows from \eqref{eq_unom2} and \eqref{eq_delt_nom1}. 
		
		First, consider the case where \ref{eq_111} is satisfied, and it follows from \eqref{eq_dleta3} that, for all $x\in\SX$, 
		\begin{align*}
			\lambda V(x) - \delta_*(x, \hat{\theta}) &\ge \lambda V(x) - \left( \frac{c_u}{\sqrt{2\rho}} + 2c_F\theta_{\max}  \right) \sqrt{V(x)},
		\end{align*}
		which, since $V$ is radially unbounded, implies that $\lim_{\|x\| \to \infty} (\lambda V(x) - \delta_*(x, \hat{\theta}))= +\infty$.
		
		Next, consider the case where \ref{eq_222} is satisfied, and it follows from \eqref{eq_dleta3} that, for all $x\in\SX$,  
		\begin{align*}
			\lambda V(x) - \delta_*(x, \hat{\theta}) &\ge (\lambda - 2c_F\theta_{\max} ) V(x) - \frac{c_u}{\sqrt{2\rho}} \sqrt{V(x)},
		\end{align*}
		which, since $V$ is radially unbounded and  $\lambda > 2c_F\theta_{\max} $, implies that  $\lim_{\|x\| \to \infty} (\lambda V(x) - \delta_*(x, \hat{\theta}))= +\infty$.
		Therefore, since in both cases \ref{eq_111} and \ref{eq_222}, $\lim_{\|x\| \to \infty} (\lambda V(x) - \delta_*(x, \hat{\theta}))= +\infty$, it follows that, for all $\varepsilon > 0$, there exists a compact level set $\mathcal{S} \subset \SX$ such that for all $x \in\SX\setminus \SSS$, $\lambda V(x) - \delta_*(x, \hat{\theta}) \ge \varepsilon$, which confirms \ref{assum:feasibility}.
	\end{proof}
	
	Building on this asymptotic dominance, the following theorem guarantees that all closed-loop signals remain uniformly bounded.
	
	\begin{theorem}[]
		\label{thm:main_results}
		Consider the closed-loop system comprising the system \eqref{eq:sys_dyn}, the parameter update law \eqref{eq:adaptation_law}, and  the adaptive control law $u_*$ defined by \eqref{eq:qp_cost_unified}--\eqref{eq:const_clf}. Assume that \ref{assum:theta}--\ref{assum:feasibility} are satisfied, and let $x(0) \in \Int\mathcal{C}$ and $\hat{\theta}(0) \in \Theta$. Let $c_\alpha>0$ satisfy $\alpha(r)\le c_\alpha r$ for all $r\ge0$, and let $\varepsilon>0$ be as in \ref{assum:feasibility}. If
		\begin{equation}
			\kappa > \frac{c_\alpha}{\varepsilon},
			\label{eq:kappa_condition}
		\end{equation}
		then all closed-loop signals $x$, $\hat{\theta}$, $u$, and $\delta$ are uniformly ultimately bounded.
	\end{theorem}
	
	\begin{proof}
		Consider the time derivative of $V_\rmc$. Since \ref{assum:theta} is satisfied, using the derivation in \eqref{eq:Vc_grouped3} implies that
		\begin{equation}
			\dot{V}_\rmc(x, \hat{\theta}) \le \frac{\alpha(h(x))}{h(x)(1+h(x))} - \kappa \Big( \lambda V(x) - \delta_*(x, \hat{\theta}) \Big).
			\label{eq_dVc}
		\end{equation}
		By hypothesis, $c_\alpha>0$ satisfies $\alpha(r)\le c_\alpha r$ for all $r\ge0$. Substituting into \eqref{eq_dVc}:
		\begin{align}
			\dot{V}_\rmc(x, \hat{\theta}) & \le \frac{c_\alpha}{1+h(x)} - \kappa \left( \lambda V(x) - \delta_*(x, \hat{\theta}) \right)\nn\\
			& \le c_\alpha - \kappa \left( \lambda V(x) - \delta_*(x, \hat{\theta}) \right),\label{eq_VCD5}
		\end{align}
		where the second inequality uses $h(x(t))>0$ for all $t\ge0$ (Theorem~\ref{thm:main_results0}), so $c_\alpha/(1+h)\le c_\alpha$.
		
		Next, note that since \ref{assum:feasibility} is satisfied, it follows from  \eqref{eq_VCD5} that, for all $x\in\mathcal{X}\setminus\SSS$  and all $\hat\theta\in\Theta$,
		\begin{equation}
			\dot{V}_\rmc(x, \hat{\theta}) \le c_\alpha - \kappa \varepsilon.\label{eq_Vdot10}
		\end{equation}
		Since \eqref{eq:kappa_condition} holds, it follows from \eqref{eq_Vdot10} that, for all $x \in \mathcal{X} \setminus \mathcal{S}$,
		\begin{equation}
			\dot{V}_\rmc(x, \hat{\theta}) \le c_\alpha - \kappa\varepsilon < 0.\label{eq_Vcdda}
		\end{equation}
		We now establish the forward invariance of a specific sublevel set. Let $s_{0} \isdef \max_{x \in \mathcal{S}, \hat{\theta} \in \Theta} V_\rmc(x, \hat{\theta})$ and define the sublevel set $\Omega_s \isdef \{(x, \hat{\theta}) \in \mathcal{X} \times \Theta \mid V_\rmc(x, \hat{\theta}) \le s\}$, where $s \ge \max\{V_\rmc(x(0), \hat{\theta}(0)), s_{0}\}$.
		Since  $s \ge s_{0}$, it follows that, for all $(x, \hat{\theta})\in \SX\setminus \Int\Omega_s$, $(x, \hat{\theta}) \in \SX\setminus\mathcal{S} \times \Theta$.
		It thus follows from \eqref_Vcdda that $\Omega_s$ is forward invariant with respect to \eqref{eq:sys_dyn}. Thus, for all $t \ge 0$, $(x(t), \hat{\theta}(t)) \in \Omega_s$, which implies $V_\rmc(t) \le s$.
		
		Finally, note that since \ref{assum:clf} implies that $V$ is radially unbounded, there exists a class $\mathcal{K}$ function $\underline{\alpha}$ such that, for all $t \ge 0$, $\underline{\alpha}(\|x(t)\|) \le V(x(t))$.
		In addition, since $\tilde\theta^\rmT\Gamma^{-1}\tilde\theta\ge0$, the definition~\eqref{eq:composite_barrier} gives $V_\rmc(x,\hat\theta)\ge V(x)$ for all $x\in\SC$ and $\hat\theta\in\Theta$. It thus follows that  $\underline{\alpha}(\|x(t)\|) \le V(x(t)) \le V_\rmc(t) \le s$ . Therefore, for all $t \ge 0$, $\|x(t)\| \le \underline{\alpha}^{-1}(s)$, which confirms that $x$ is uniformly bounded.
		Since, in addition, Lemma \ref{lem:param_boundedness} implies, for all $t \ge 0$,  $\hat{\theta}(t)\in\Theta$, it follows that the joint state $(x(t), \hat{\theta}(t))$ evolves within a compact domain.
		Therefore, since the image of a compact set under a continuous map is compact, part \ref{Thm2_p2} of Theorem \ref{thm:feasibility} implies that  the control input $u_*$ and relaxation $\delta_*$ are uniformly bounded, which confirms the result.
	\end{proof}
	
	\begin{remark}
		\label{rem:uub_bound}
		Theorem~\ref{thm:main_results} establishes that all closed-loop signals are uniformly ultimately bounded within the sublevel set $\Omega_s$, where $s\ge\max\{V_{\rm c}(x(0),\hat\theta(0)),s_0\}$. The quantity $s_0$ depends on the maximum of $V_{\rm c}$ over the compact set $\mathcal{S}$ (determined by $\varepsilon$ from~\ref{assum:feasibility} and the system dynamics), while the initial value $V_{\rm c}(x(0),\hat\theta(0))$ reflects the starting conditions. Consequently, a large initial parameter error $\|\tilde\theta(0)\|$ increases $V_{\rm c}(x(0),\hat\theta(0))$ and therefore increases $s$. In practice, better initial guesses $\hat\theta(0)$ (from offline identification, for example) directly reduce the size of the invariant set and improve transient behavior. The relaxation variable $\delta_*$ is also bounded on compact sets; however, the bound on $\delta_*$ grows with the degree of safety-stability conflict and the magnitude of $\tilde\theta$.
	\end{remark}
	
	Lemma \ref{prop:asymptotic_dominance} establishes that for large states, the drive for stability outpaces the relaxation penalties required for safety. Under the conditions of Lemma~\ref{prop:asymptotic_dominance}, the required gain $\kappa$ in \eqref{eq:kappa_condition} can be computed explicitly, as shown in the following corollary.
	
	\begin{corollary}
		\label{cor:kappa_explicit}
		Assume that  \ref{assum:theta}--\ref{assum:gains} and \ref{assum:clf} are satisfied, and let $\lambda>0$ and $u_\rmn\colon \SX \to \BBR^m$ be such that, for all $x\in\SX$, \eqref{eq_prop1V} is satisfied.
		In addition, assume that there exists $c_{u} > 0$ such that, for all $x \in \SX$, \eqref{eq_unom2} is satisfied. 
		Furthermore, assume that there exists $c_{F} > 0$ such that condition~\ref{eq_111} of Lemma \ref{prop:asymptotic_dominance} is satisfied.  
		Let $\alpha(r)=\alpha_0 r$, and define 
		\begin{equation}
			A \isdef \frac{c_u}{\sqrt{2\rho}} + 2c_F\theta_{\max}.\label{eq:A_def}
		\end{equation}
		Then, for all $\kappa > \frac{\alpha_0\lambda}{2A^2},$ \eqref{eq:kappa_condition} is satisfied.
	\end{corollary}
	
	\begin{proof}
		Using the same process as in the proof of Lemma \ref{prop:asymptotic_dominance} confirms \eqref{eq_dleta3}, which together with condition~\ref{eq_111} of Lemma \ref{prop:asymptotic_dominance} implies that, for all $x\in\SX$,
		\begin{equation}
			\delta_*(x,\hat\theta)\le A\sqrt{V(x)}. \label{eq_cor_delta}
		\end{equation}
		Define $g\colon[0,\infty)\to\BBR$ by $g(r)\isdef \lambda r^2 - A r$. Using $r=\sqrt{V(x)}$, it follows from \eqref{eq_cor_delta} that
		\begin{align}
			\lambda V(x) - \delta_*(x,\hat\theta) &\ge \lambda V(x) - A\sqrt{V(x)} \nn\\
			&= g\!\left(\sqrt{V(x)}\right). \label{eq_cor_g}
		\end{align}
		Since $g'(r)=2\lambda r - A$, the function $g$ is  increasing on $(A/(2\lambda),\infty)$.
		Define the compact set $\mathcal{S}\isdef\{x\in\SX: V(x)\le 4(A/\lambda)^2\}$, and note that, for all $x\notin\mathcal{S}$, $\sqrt{V(x)}\ge 2A/\lambda > A/(2\lambda)$.
		Since $g$ is increasing on this interval, its minimum over $\{r\ge 2A/\lambda\}$ is attained at $r=2A/\lambda$, that is, for all $x\notin\mathcal{S}$,
		\begin{equation}
			g\!\left(\sqrt{V(x)}\right) \ge g\!\left(\frac{2A}{\lambda}\right)
			= \frac{2A^2}{\lambda}. \label{eq_cor_min}
		\end{equation}
		Combining \eqref{eq_cor_g} and \eqref{eq_cor_min} implies that, for all $x\notin\mathcal{S}$ and all $\hat\theta\in\Theta$, $\lambda V(x) - \delta_*(x,\hat\theta) \ge 2A^2/\lambda$.
		Hence \ref{assum:feasibility} holds with $\varepsilon = 2A^2/\lambda$. Since $\alpha(r)=\alpha_0 r$ gives $c_\alpha=\alpha_0$, the gain condition \eqref{eq:kappa_condition} implies $\kappa > \alpha_0\big/(2A^2/\lambda) = \alpha_0\lambda/(2A^2)$, which is the stated bound.
	\end{proof}
	
	\begin{remark}
		Corollary~\ref{cor:kappa_explicit} makes the gain condition explicit and computable from design parameters: $c_u$ (i.e., nominal control bound~\eqref{eq_unom2}), $\rho$ (i.e., stability-relaxation penalty), $c_F$ (i.e., regressor bound from condition~\ref{eq_111}), and $\theta_{\max}$ (i.e., parameter bound from \ref{assum:theta}). Notably, $\kappa$ does not appear in Theorems~\ref{thm:main_results0} or~\ref{thm:noise_robust}: forward invariance holds for any $\kappa>0$. The lower bound $\kappa > \frac{\alpha_0\lambda}{2A^2}$ is needed only for uniform boundedness of signals.
	\end{remark}
	
	The adaptation law \eqref{eq:adaptation_law} contains the term $-\psi(x)^\rmT/[h(x)(1+h(x))]$, which is unbounded as $h(x)\downarrow0$, and the projection \eqref{eq:proj_def} constrains $\hat\theta$ but not $\dot{\hat\theta}$. The following result shows that $\dot{\hat\theta}$ is nevertheless uniformly bounded. The mechanism is that the logarithmic barrier in \eqref{eq:composite_barrier} is self-regulating: invariance of the sublevel set $\Omega_s$ established in Theorem~\ref{thm:main_results} forces a uniform positive lower bound on $h$ along closed-loop trajectories, which in turn bounds the barrier gradient.
	
	\begin{corollary}
		\label{cor:thetadot_bounded}
		Consider the closed-loop system of Theorem~\ref{thm:main_results}, and assume that \ref{assum:theta}--\ref{assum:feasibility} and \eqref{eq:kappa_condition} are satisfied. Let $s>0$ and $\Omega_s$ be as in the proof of Theorem~\ref{thm:main_results}. Then, the following statements hold:
		\begin{enumerate}[label=\roman*)]
			\item \label{cor7_p1} For all $t\ge0$, $h(x(t))\ge \underline{h},$ where $\underline{h}\isdef\dfrac{e^{-s}}{1-e^{-s}}>0$.
			\item \label{cor7_p2} There exists $c_\theta>0$ such that, for all $t\ge0$, $\|\dot{\hat\theta}(t)\le c_\theta$.
		\end{enumerate}
	\end{corollary}
	
	\begin{proof}
		To prove \ref{cor7_p1}, note that Theorem~\ref{thm:main_results} implies that, for all $t\ge0$, $V_\rmc(x(t),\hat\theta(t))\le s$.
		Since, for all $x\in\SX$, $\kappa V(x)\ge0$ and $\tfrac12\tilde\theta^\rmT\Gamma^{-1}\tilde\theta\ge0$, it follows from \eqref{eq:composite_barrier} that, for all $t\ge0$,
		\begin{equation}
			-\ln\!\left(\frac{h(x(t))}{1+h(x(t))}\right)\le V_\rmc(x(t),\hat\theta(t)) \le s,
		\end{equation}
		which implies $h(x(t))/(1+h(x(t)))\ge e^{-s}$.
		Since $h(x(t))>0$ by Theorem~\ref{thm:main_results0} and $e^{-s}\in(0,1)$, rearranging yields $h(x(t))(1-e^{-s})\ge e^{-s}$, which confirms \ref{cor7_p1}.
		
		To prove \ref{cor7_p2}, note that Theorem~\ref{thm:main_results} implies that $x(t)$ evolves in the compact set $\SX_s\isdef\{x\in\SX\colon \|x\|\le \underline{\alpha}^{-1}(s)\}$, and Lemma~\ref{lem:param_boundedness} implies $\hat\theta(t)\in\Theta$.
		Since $\psi$, $\phi$, and $F$ are continuous, there exist $\bar\psi,\bar\phi,\bar F>0$ such that, for all $x\in\SX_s$, $\|\psi(x)\|\le\bar\psi$, $\|\phi(x)\|\le\bar\phi$, and $\|F(x)\|\le\bar F$.
		Moreover, since $\hat\theta,\theta_*\in\Theta$, it follows that $\|\tilde\theta\|\le2\theta_{\max}$, and thus \eqref{eq_err_sub} implies $\|e(x,\hat\theta)\|\le 2\bar F\theta_{\max}$.
		Let $\tau$ denote the unprojected update direction in \eqref{eq:adaptation_law}. Using \ref{cor7_p1} and the triangle inequality yields that, for all $t\ge0$, $\|\tau\|\le \bar{\tau},$ where
		\begin{equation}
			\bar\tau\isdef \lambda_{\max}(\Gamma)\left(\kappa\bar\phi + \frac{\bar\psi}{\underline{h}(1+\underline{h})} + 2\gamma\bar F^2\theta_{\max}\right)
			\label{eq:tau_bound}
		\end{equation}
		If the projection is inactive, then $\dot{\hat\theta}=\tau$ and $\|\dot{\hat\theta}\|\le\bar\tau$.
		If the projection is active, then $\|\hat\theta\|=\theta_{\max}$, and it follows from \eqref{eq:proj_def} that
		\begin{equation}
			\left\|\Gamma\frac{\hat\theta^\rmT\tau}{\hat\theta^\rmT\Gamma\hat\theta}\hat\theta\right\|
			\le  \frac{\lambda_{\max}(\Gamma)}{\lambda_{\min}(\Gamma)}\|\tau\|,
		\end{equation}
		and thus $\|\dot{\hat\theta}\|\le\big(1+\lambda_{\max}(\Gamma)/\lambda_{\min}(\Gamma)\big)\bar\tau$, which confirms \ref{cor7_p2} with $c_\theta\isdef\big(1+\lambda_{\max}(\Gamma)/\lambda_{\min}(\Gamma)\big)\bar\tau$.
	\end{proof}
	
	\begin{remark}
		\label{rem:thetadot_interpretation}
		Corollary~\ref{cor:thetadot_bounded} shows that the apparent singularity of the barrier gradient in \eqref{eq:adaptation_law} is never encountered along closed-loop trajectories. The bound $\underline{h}$ degrades as $s$ increases, so larger initial energy $V_\rmc(x(0),\hat\theta(0))$, which includes larger initial parameter error, permits closer approach to the boundary and correspondingly larger $\|\dot{\hat\theta}\|$. This is consistent with the intended behavior of the design: proximity to the safety boundary accelerates adaptation, and Corollary~\ref{cor:thetadot_bounded} certifies that this acceleration remains bounded.
	\end{remark}
	
	\begin{definition}
		\label{def:pe}
		The regressor $F$ is \emph{persistently exciting} (PE) along the trajectory of \eqref{eq:sys_dyn} if there exist constants $T_0,\mu>0$ such that, for all $t\ge0$,
		\begin{equation}
			\int_t^{t+T_0} F(x(\tau))^\rmT F(x(\tau))\,\rmd\tau \succeq \mu I_p. \label{eq:pe_condition}
		\end{equation}
	\end{definition}
	
	Although the main results guarantee safety without persistence of excitation, it is worth recording what happens when the trajectory \emph{is} informative, since this directly addresses the practical concern that the relaxation $\delta_*$ may remain large. If $F$ is persistently exciting, the gradient term $\gamma F(x)^\rmT e$ in \eqref{eq:adaptation_law}  drives the parameter error to zero, which in turn tightens the bound \eqref{eq_dleta3} on $\delta_*$. The following result formalizes this. Its proof follows from the standard exponential-stability theory for persistently excited gradient estimators \cite[Sec.~4.3]{ioannou2012robust}, applied once the projection \eqref{eq:proj_def} becomes inactive.
	
	\begin{proposition}
		\label{prop:pe_convergence}
		In addition to the assumptions of Theorem~\ref{thm:main_results}, assume that the following conditions hold:
		\begin{enumerate}[label=(\roman*)]
			\item \label{pe_strict} $\|\theta_*\| < \theta_{\max}$.
			\item \label{pe_cond}    $F$ is PE along the closed-loop trajectory with constants $T_0,\mu>0$.
		\end{enumerate}
		Then, there exist constants $\beta,\sigma>0$ such that, for all $t\ge0$,
		\begin{equation}
			\|\tilde\theta(t)\|^2 \le \beta \|\tilde\theta(0)\|^2 e^{-\sigma t}. \label{eq:pe_convergence}
		\end{equation}
	\end{proposition}
	
	\begin{remark}
		\label{rem:pe_examples}
		Proposition~\ref{prop:pe_convergence} shows that persistent excitation upgrades the safety guarantee of Theorem~\ref{thm:main_results0} with exponential parameter convergence. As $\tilde\theta\to0$, the bound \eqref{eq_dleta3} on the relaxation variable $\delta_*$ tightens, so the closed-loop behavior approaches that of the exact-model controller. This complements the main results, which guarantee safety for all bounded uncertainties \emph{without} requiring PE. The two regimes are illustrated in Section~\ref{sec:simulation}: in Example~\ref{ex_Omni}, the curved avoidance trajectory excites both drift components, so $\hat\theta$ converges (Figure~\ref{fig:omni_results}); in Example~\ref{ex_ACC}, the nearly constant-velocity profile is not persistently exciting, so $\hat\theta$ does not converge, yet safety holds throughout.
	\end{remark}
	
	Next, we consider the case where the state-derivative measurements are noisy. 
	Let $\dot x_{\rm m}\colon[0,\infty)\to\BBR^n$ denote the noisy measurement of $\dot x$ defined by
	\begin{equation*}
		\dot x_{\rm m}(t) \isdef \dot x(t) + w(t)
	\end{equation*}
	where the noise $w\colon[0,\infty)\to\BBR^n$ is measurable, and there exists $\bar w\ge0$ such that, for all $t\ge0$,  $\|w(t)\| \le \bar w$.
	Let $e_{\rm m}(x,\hat\theta) \isdef \dot x_{\rm m} - (f(x)+F(x)\hat\theta+G(x)u)$ be the noisy estimation error computed from $\dot x_{\rm m}$.
	
	The following theorem strengthens Theorem~\ref{thm:main_results0} by showing that the safety guarantee is robust to bounded errors in the state-derivative measurement.
	
	\begin{theorem}
		\label{thm:noise_robust}
		Consider the closed-loop system comprising \eqref{eq:sys_dyn}, the adaptation law \eqref{eq:adaptation_law} implemented with $e_{\rm m}$ in place of $e$, and the adaptive control law $u_*$ defined by \eqref{eq:qp_cost_unified}--\eqref{eq:const_clf}.
		Assume that \ref{assum:theta}--\ref{assum:clf} are satisfied, and let $x(0)\in\Int\mathcal{C}$ and $\hat\theta(0)\in\Theta$.
		Then, $\mathcal{C}$ is forward invariant with respect to \eqref{eq:sys_dyn}.
	\end{theorem}
	
	\begin{proof}
		Since $\dot x_{\rm m} = \dot x + w$, the noisy estimation error satisfies
		\begin{equation}
			e_{\rm m}(x,\hat\theta) = -F(x)\tilde\theta + w(t). \label{eq:noisy_e}
		\end{equation}
		We use the same process as in the proof of Theorem \ref{thm:main_results0}. Repeating the derivation of \eqref{eq:Vc_grouped2} verbatim but with $e_{\rm m}$ in place of $e$ in the adaptation law~\eqref{eq:adaptation_law}, and applying inequality~\eqref{eq_proj_property} of Lemma~\ref{lem:proj_property} in the same manner, yields
		\begin{align}
			\dot V_{\rm c}(x,\hat\theta)
			&\le -\frac{\dot h_{\hat\theta}(x,\hat\theta)}{h(x)(1+h(x))} + \kappa\dot V_{\hat\theta}(x,\hat\theta) \nn\\
			&\quad + \gamma\tilde\theta^\rmT F(x)^\rmT e_{\rm m}(x,\hat\theta). \label{eq:Vc_noisy_pre}
		\end{align}
		Using \eqref{eq:noisy_e}, it follows that
		\begin{align}
			\gamma\tilde\theta^\rmT F(x)^\rmT e_{\rm m}(x,\hat\theta)
			&= -\gamma\|F(x)\tilde\theta\|^2 + \gamma\tilde\theta^\rmT F(x)^\rmT w \nn\\
			&\le -\gamma\|e(x,\hat\theta)\|^2 + \gamma\|\tilde\theta\|\|F(x)\|\bar w, \label{eq:cross_term_noisy}
		\end{align}
		where $e(x,\hat\theta)\isdef-F(x)\tilde\theta$ is the noise-free error.
		Since $\hat\theta,\theta_*\in\Theta$, it follows that $\|\tilde\theta\|\le 2\theta_{\max}$.
		Combining \eqref{eq:Vc_noisy_pre} and \eqref{eq:cross_term_noisy} and substituting the QP constraints \eqref{eq:const_cbf}--\eqref{eq:const_clf} yields
		\begin{align}
			\dot V_{\rm c}(x,\hat\theta) &\le \frac{\alpha(h(x))}{h(x)(1+h(x))} - \kappa\lambda V(x) - \gamma\|e\|^2 \nn\\
			&\quad + \kappa\delta_*(x,\hat\theta) + c_w, \label{eq:Vc_noisy}
		\end{align}
		where $c_w$ bounds the noise contribution $\gamma\|\tilde\theta\|\|F(x)\|\bar w$ and is defined below.
		Suppose for contradiction that there exists $t_1>0$ such that $h(x(t_1))=0$ and, for all $t\in[0,t_1)$, $h(x(t))>0$.
		Since $x$ is continuous on $[0,t_1]$, it belongs to a compact set. It thus follows that there exists $\bar c_F>0$ such that, for all $t\in[0,t_1]$, $\|F(x(t))\|\le \bar c_F$  and thus the term $c_w \isdef 2\gamma\theta_{\max}\bar c_F\bar w$ in \eqref{eq:Vc_noisy} is finite on $[0,t_1]$.
		It thus follows from \eqref{eq:composite_barrier} that $\lim_{t\uparrow t_1}V_{\rm c}(x(t),\hat\theta(t))=+\infty$.
		Since $\alpha$ is Lipschitz with $\alpha(0)=0$, there exists $c_\alpha>0$ such that $\alpha(h)\le c_\alpha h$, and thus $\alpha(h)/[h(1+h)]\le c_\alpha$.
		Omitting the nonpositive terms $-\kappa\lambda V$ and $-\gamma\|e\|^2$ from \eqref{eq:Vc_noisy} yields that, for all $t\in[0,t_1]$, 
		\begin{equation}
			\dot V_{\rm c}(x(t),\hat\theta(t)) \le c_0
		\end{equation}
		where $c_0\isdef c_\alpha + \kappa\delta_{\max} + c_w$ and $\delta_{\max}\isdef\sup_{\tau\in[0,t_1]}\delta_*(x(\tau),\hat\theta(\tau))<\infty$ by part~\ref{Thm2_p2} of Theorem~\ref{thm:feasibility} and the Lipschitz continuity of $\delta_*$.
		Integrating over $[0,t_1]$ implies $V_{\rm c}(x(t_1),\hat\theta(t_1))\le V_{\rm c}(x(0),\hat\theta(0))+t_1 c_0<+\infty$, which contradicts $\lim_{t\uparrow t_1}V_{\rm c}(x(t),\hat\theta(t))=+\infty$.
		Therefore, for all $t\ge0$, $h(x(t))>0$. 
	\end{proof}
	
	\begin{remark}
		\label{rem:noise_uub}
		Theorem~\ref{thm:noise_robust} shows that the forward invariance guarantee of CaCBF is robust to bounded state-derivative errors. Note that Theorem~\ref{thm:noise_robust} does not require $\bar w$ to be sufficiently small.
		The UUB result of Theorem~\ref{thm:main_results} extends to the noisy case with the modified condition $\kappa>(c_\alpha+c_w)/\varepsilon$. Thus, using Lemma~\ref{prop:asymptotic_dominance} with condition~\ref{eq_111}, so that $\varepsilon=2A^2/\lambda$ and $c_\alpha=\alpha_0$, it follows that $\kappa > (\alpha_0 + c_w)\lambda/(2A^2)$; a larger noise level $\bar w$ increases $c_w$ and therefore requires a larger $\kappa$.
	\end{remark}
	
	\begin{remark}
		\label{rem:pe_noise}
		Proposition~\ref{prop:pe_convergence} is stated for the noise-free adaptation law \eqref{eq:adaptation_law}. Under the bounded state-derivative error of Theorem~\ref{thm:noise_robust}, the estimation error satisfies $e_{\rm m}=-F(x)\tilde\theta+w$. Thus, the parameter-error dynamics acquire the additive perturbation $\gamma\Gamma F(x)^\rmT w$, which is bounded by $\gamma\lambda_{\max}(\Gamma)\bar F\bar w$ on $\Omega_s$. Since the persistently excited nominal dynamics are exponentially stable, this perturbation is input-to-state stable, and consequently $\tilde\theta$ converges not to zero but to a residual ball, that is, there exists $c_{\rm r}>0$ such that $\limsup_{t\to\infty}\|\tilde\theta(t)\|\le c_{\rm r}\bar w$, with $c_{\rm r}$ inversely proportional to the convergence rate $\sigma$. Exact parameter convergence therefore requires $\bar w=0$, whereas the safety guarantee of Theorem~\ref{thm:noise_robust} holds for every $\bar w\ge0$.
	\end{remark}
	
	\section{Comparison with Robust CBF}
	
	This section quantifies the advantage of the CaCBF framework in two ways: a structural comparison of the adaptive and robust admissible control sets, and an architectural comparison with representative prior adaptive CBF methods. Let $\theta_\rme \in \BBR^p$ represent a fixed parameter estimate. We define the worst-case uncertainty margin $\sigma_{\rm rob}$ as the maximum possible deviation over the parameter set given by
	\begin{equation}
		\sigma_{\rm rob}(x, {\theta}_\rme) \isdef \max_{\vartheta \in \Theta} \left| L_F h(x) (\vartheta - {\theta}_\rme) \right|.\label{eq_sisgam}
	\end{equation}
	A standard robust CBF (R-CBF) guarantees safety by guarding against the worst-case parameter realization by imposing the  conservative constraint
	\begin{align}
		L_f h(x) + \psi(x) {\theta}_\rme &+L_G h(x) u - \sigma_{\rm rob}(x, {\theta}_\rme)\geq - \alpha(h(x)).
		\label{eq:robust_cbf}
	\end{align}
	In contrast, the proposed CaCBF uses the evolving estimate $\hat{\theta}$ to enforce the nominal constraint \eqref{eq:const_cbf}. By swapping the static worst-case margin for a dynamic adaptation process, the CaCBF expands the feasible control space. The following theorem confirms this advantage: it proves that the admissible control set of the adaptive formulation contains the robust set as a subset.
	
	\begin{theorem} \label{prop:reduced_conservatism}
		Assume that \ref{assum:theta} is satisfied, let $\theta_\rme\in{\BBR^p}$ and $\hat{\theta}\in{\Theta}$. Furthermore, let $\SU_{\rm rob}(x, {\theta_\rme})$ and $\SU_{\rm adp}(x, \hat{\theta})$ denote the sets of admissible control inputs satisfying the R-CBF condition \eqref{eq:robust_cbf} and the CaCBF condition \eqref{eq:const_cbf}, respectively. Then, for all $x \in \mathcal{C}$, 
		\begin{equation}
			\SU_{\rm rob}(x, {\theta_\rme}) \subseteq \SU_{\rm adp}(x, \hat{\theta}).\label{eq_rob_adap}
		\end{equation}
	\end{theorem}
	
	\begin{proof}
		Let $x\in\SC$ and $u \in \SU_{\rm rob}(x, {\theta_\rme})$. By definition, $u$ satisfies the robust inequality
		\begin{align}
			L_G h(x) u &\geq -L_f h(x) - \psi(x) {\theta_\rme}  - \alpha(h(x)) + \sigma_{\rm rob}(x, {\theta_\rme}).\label{eq_SIGMA_ROB_PROOF}
		\end{align}
		Since $\hat{\theta} \in \Theta$, it follows from the definition of the worst-case margin that $\sigma_{\rm rob}(x, {\theta_\rme}) \geq \big|\psi(x)(\hat{\theta} - \theta_\rme)\big| \geq \psi(x)(\theta_\rme - \hat{\theta})$. It thus follows from \eqref{eq_SIGMA_ROB_PROOF} that
		\begin{align}
			L_G h(x) u &\geq -L_f h(x) - \psi(x) \theta_\rme - \alpha(h(x)) + \psi(x)(\theta_\rme - \hat{\theta}) \nonumber \\
			&= -L_f h(x) - \psi(x) \hat{\theta} - \alpha(h(x)),
		\end{align}
		which matches \eqref{eq:const_cbf}. Thus, $u \in \SU_{\rm adp}(x, \hat{\theta})$, which confirms the inclusion \eqref{eq_rob_adap}.
	\end{proof}
	
	Equation \eqref{eq_rob_adap} confirms that CaCBF recovers control authority that the R-CBF approach would otherwise discard. R-CBF shrinks the safe set, guarding against the worst possible parameter configuration at every instant. In contrast, CaCBF uses the current parameter estimate. Rather than relying on a static, worst-case margin, CaCBF relies on the adaptation mechanism to correct errors online. This allows the system to compensate for uncertainty dynamically, granting access to a larger set of safe control inputs.
	
	\begin{remark}
		\label{rem:prior_comparison}
		The CaCBF shares the broad goal of combining adaptation with safety but
		differs from existing methods in key structural ways. First, note that \cite{taylor2020adaptive} introduced 
		an adaptive CBF by coupling a safety-gradient adaptation law to the
		barrier constraint. Their formulation has no CLF component and no composite
		energy, whereas CaCBF includes both, producing the prediction-error gradient
		$\gamma F(x)^\rmT e(x,\hat\theta)$ in \eqref{eq:adaptation_law} and the PE
		convergence result of Proposition~\ref{prop:pe_convergence}. 
		On the other hand, \cite{lopez2023unmatched} extended adaptive
		CBFs to the certainty-equivalence setting with unmatched uncertainties,
		using the current estimate directly without a Lyapunov-based design for the
		adaptation law. In contrast, CaCBF derives the update law from the composite
		energy $V_{\rm c}$, which simultaneously shapes the estimation error, the CLF
		decay, and the barrier gradient into a single dissipation inequality. In addition, \cite{xiao2022adaptive} employed dynamic penalty
		functions to soften the barrier constraint when the system is
		over-constrained, but their stability analysis relies on the assumption that
		the penalty converges. CaCBF avoids this by using the CLF relaxation variable
		$\delta_*$, whose boundedness is established without parameter convergence
		or PE. Furthermore, \cite{cohen2022high} addressed high-order
		relative-degree systems via adaptive high-order CBFs; their construction is
		complementary to CaCBF, and extending CaCBF to the high-order case is a
		direction for future work.
		
		Table~\ref{tab:comparison} summarizes these distinctions across eight
		capability dimensions. The primary distinction of CaCBF is the
		composite Lyapunov design that tightly couples safety, stability, and
		identification into a single energy function, yielding all four columns in the
		right portion of the table simultaneously.
	\end{remark}
	
	\begin{table*}[t]
		\centering
		\small
		\caption{Capability comparison of adaptive CBF methods. \cmark~= formally
			established; \pmark~= holds under additional assumption (PE or penalty
			convergence); \xmark~= not claimed. \emph{Safe (PE-free)}: forward invariance
			without persistence of excitation. \emph{CLF integrated}: stability objective
			built into the QP constraint. \emph{Composite $V_{\rm c}$}: safety, stability,
			and estimation unified in a single Lyapunov function. \emph{Adapt.\ from
				$V_{\rm c}$}: adaptation law derived by differentiating $V_{\rm c}$, not
			appended separately. \emph{UUB (no PE)}: uniform ultimate boundedness without
			PE. \emph{PE conv.}: exponential parameter convergence under PE.
			\emph{Noise-robust}: formal safety guarantee under bounded $\dot x$
			measurement error. \emph{Explicit $\kappa$}: computable sufficient condition
			on the barrier weight.}
		\label{tab:comparison}
		\begin{tabular*}{\textwidth}{@{\extracolsep\fill}lcccc cccc@{}}
			\toprule
			& \multicolumn{4}{c}{\textit{Safety \& Adaptation Architecture}}
			& \multicolumn{4}{c}{\textit{Formal Guarantees}} \\
			\cmidrule(lr){2-5}\cmidrule(lr){6-9}
			\textbf{Method}
			& \makecell{Safe\\(PE-free)}
			& \makecell{CLF\\integrated}
			& \makecell{Composite\\$V_{\rm c}$}
			& \makecell{Adapt.\\from $V_{\rm c}$}
			& \makecell{UUB\\(no PE)}
			& \makecell{PE\\conv.}
			& \makecell{Noise-\\robust}
			& \makecell{Explicit\\$\kappa$} \\
			\midrule
			Taylor \& Ames~\cite{taylor2020adaptive}
			& \cmark & \xmark & \xmark & \xmark
			& \xmark & \xmark & \xmark & \xmark \\
			Lopez \& Slotine~\cite{lopez2023unmatched}
			& \cmark & \xmark & \xmark & \xmark
			& \xmark & \xmark & \xmark & \xmark \\
			Xiao et al.~\cite{xiao2022adaptive}
			& \pmark & \cmark & \xmark & \xmark
			& \pmark & \xmark & \xmark & \xmark \\
			Cohen et al.~\cite{cohen2022high}
			& \cmark & \cmark & \xmark & \xmark
			& \xmark & \xmark & \xmark & \xmark \\
			\textbf{CaCBF (proposed)}
			& \cmark & \cmark & \cmark & \cmark
			& \cmark & \cmark & \cmark & \cmark \\
			\bottomrule
		\end{tabular*}
	\end{table*}
	
	\begin{remark}
		\label{rem:thm5_interp}
		Theorem~\ref{prop:reduced_conservatism} holds for \emph{all} $\hat\theta\in\Theta$, including the worst-case estimate, thus the inclusion $\SU_{\rm rob}\subseteq\SU_{\rm adp}$ is a structural property of the constraint geometry rather than a quantification of how much the estimate has improved. In particular, the practical significance of Theorem~\ref{prop:reduced_conservatism} is twofold. First, it guarantees that the adaptive QP is \emph{always} at least as feasible as the robust one: any scenario in which R-CBF produces a feasible solution is also feasible for CaCBF, but not vice versa (See  Example~\ref{ex_Drone_Gate} where R-CBF fails entirely). Second, the gap between $\SU_{\rm rob}$ and $\SU_{\rm adp}$ is determined by $\sigma_{\rm rob}(x,\theta_\rme)-|\psi(x)(\hat\theta-\theta_\rme)|$; this gap  grows as $\hat\theta$ moves away from the worst-case direction, quantifying the performance advantage of a given estimate. The numerical metrics (e.g., $h_{\min}$ and $\eta$) reported in Section~\ref{sec:simulation} provide empirical evidence of this advantage for specific trajectories.
	\end{remark}
	
	Figure \ref{fig:set_containment} provides a geometric interpretation of this result. It illustrates that, for each $x\in\SC$, $\theta_\rme\in{\BBR^p}$, and $\hat{\theta}\in\Theta$, the robust safe control set $\SU_{\rm rob}(x, \theta_\rme)$, computed using the fixed estimate $\theta_\rme$, is nested within the adaptive safe control set $\SU_{\rm adp}(x, \hat{\theta})$. As the estimate $\hat{\theta}$ converges to the true parameter, the dashed boundary of $\SU_{\rm adp}(x, \hat{\theta})$ expands to recover the true safe control set $\SU_{\rm cbf}(x)$. Equation \eqref{eq_rob_adap} confirms that CaCBF recovers control authority that the robust formulation discards.
	
	\begin{figure}[t!]
		\centering
		\definecolor{myb}{HTML}{0071BC}
		\definecolor{myr}{HTML}{ED135A}
		\definecolor{myn}{HTML}{696969}
		
		\begin{tikzpicture}[scale=0.82]
			\fill[fill=myn!5] (0,0) ellipse (4.2cm and 2.6cm);
			\draw[color=black, line width=1.0pt] (0,0) ellipse (4.2cm and 2.6cm);
			\node[color=black] at (0, 2.0) { $\SU_{\rm cbf}(x)$};
			
			\fill[fill=myr!10] (0,-0.3) ellipse (3.0cm and 1.8cm);
			\draw[dashed, color=myr, line width=1.2pt] (0,-0.3) ellipse (3.0cm and 1.8cm);
			\node[color=myr] at (0, 0.7) { $\SU_{\rm adp}(x, \hat{\theta})$};
			
			\fill[fill=myb!10] (0,-0.6) ellipse (1.6cm and 0.9cm);
			\draw[color=myb, line width=1.2pt] (0,-0.6) ellipse (1.6cm and 0.9cm);
			\node[color=myb] at (0, -0.6) { $\SU_{\rm rob}(x, \theta_\rme)$};
		\end{tikzpicture}
		\caption{Geometric interpretation of Theorem \ref{prop:reduced_conservatism} at each $x\in\SC$, $\theta_\rme\in{\BBR^p}$, and $\hat{\theta}\in\Theta$. The conservative robust safe control set $\SU_{\rm rob}(x, \theta_\rme)$ (inner, blue) is nested within the adaptive safe control set $\SU_{\rm adp}(x, \hat{\theta})$ (middle, red, dashed). The adaptive set expands toward the true safe control set $\SU_{\rm cbf}(x)$ (outer, solid) as $\hat{\theta}$ converges to the true parameter $\theta_*$.}
		\label{fig:set_containment}
	\end{figure}
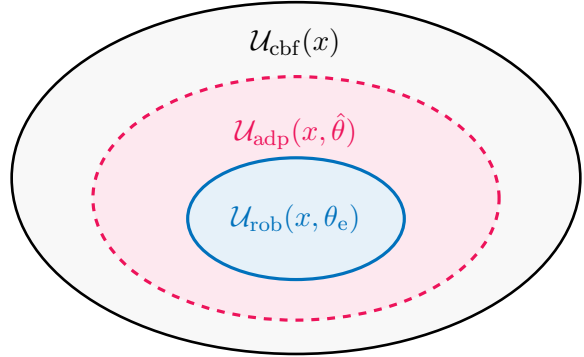
	
	\section{Numerical Examples}
	\label{sec:simulation}
	
	We illustrate the proposed CaCBF framework through numerical simulations on three distinct dynamical systems. First, we use an adaptive cruise control problem to illustrate and quantify the benefits of adaptation, showing how the CaCBF recovers control authority that a standard R-CBF would discard. We then broaden the scope to include an omnidirectional robot and a planar drone, demonstrating that the safety guarantees hold across diverse uncertain physical systems.
	
	\begin{example}\label{ex_ACC}
		\textit{Adaptive cruise control with unknown resistance force.}
		We illustrate the framework using an adaptive cruise control problem based on the point-mass model in \cite{ames2016control}. The ego vehicle must follow a lead vehicle that executes a variable-velocity profile 
		\begin{equation}
			v_\ell(t)\isdef \begin{cases}
				10,& ~0\le t<25,\\
				14,&25\le t<50,\\
				18,&50\le t<75,\\
				15,&75\le t<100,
			\end{cases}
		\end{equation}
		while maintaining a desired cruising speed $v_\rmd \isdef 26$ m/s and a safe time headway $\tau_\rmd \isdef 1.8$ s. We define the safety set $\mathcal{C}$ via the zeroing barrier function
		\begin{equation}
			h(x) = d - 1.8 v,
		\end{equation} 
		where $v$ is the ego velocity and $d$ is the following distance. Let the state be $x = [ v~\, d] ^\rmT \in \BBR^2$.  The system evolves according to
		\begin{align}
			\dot{x} &= \begin{bmatrix} 0 \\ v_\ell - v \end{bmatrix} + \begin{bmatrix} -\frac{1}{\overline{m}} & -\frac{v}{\overline{m}} & -\frac{v^2}{\overline{m}} \\ 0 & 0 & 0 \end{bmatrix} \theta_*  + \begin{bmatrix} \frac{1}{\overline{m}} \\ 0 \end{bmatrix} u,
			\label{eq:acc_dynamics}
		\end{align}
		where $\overline{m} \isdef 1650$ kg is the mass, $u$ is the wheel force, and $\theta_* = \matl 7& 6& 5\matr^\rmT$ captures the unknown coefficients for rolling resistance and aerodynamic drag.
		The initial conditions are $v(0)=25$~m/s and $d(0)=60$~m.
		To respect traction limits and passenger comfort, we bound the control input such that, for all $t\ge0$,
		\begin{equation}
			-0.25 \overline{m} g \leq u(t) \leq 0.25 \overline{m} g,
			\label{eq:input_constraints}
		\end{equation}
		where $g=9.81\,{\rm m/s^2}$ is the gravitational acceleration.
		
		We first implement the proposed CaCBF. We select the CLF candidate $V(x) = (v - v_\rmd)^2$ with a nominal convergence rate $\lambda = 2$.
		Inserting the dynamics \eqref{eq:acc_dynamics} and the Lie derivatives calculated with the estimate $\hat{\theta}$ into the  CLF-CBF-QP \eqref{eq:qp_cost_unified}--\eqref{eq:const_clf} yields, for all $x \in \BBR^n$ and $\hat{\theta} \in \BBR^p$, the  optimization
		\begin{align}
			\big({u}_*(x,\hat\theta), \delta_*(x,\hat\theta)\big) \nn\\
			&\hspace{-3.2em}= \operatorname*{argmin}_{(u, \delta)\in\BBR\times[0,\infty)} \ \tfrac{1}{2} u^2 + \rho \delta^2, \\
			\text{s.t.} \quad & v_\ell - v + \frac{1.8\big(\zeta(v)^\rmT \hat{\theta}- u\big)}{\overline{m}} \geq 1.8 v-d, \label{eq_safeQP} \\
			& \frac{2( v_\rmd-v)\big(\zeta(v)^\rmT \hat{\theta} - u\big)}{\overline{m}} \leq -2 (v - v_\rmd)^2 + \delta,
		\end{align}
		where the regressor $\zeta(v) \isdef [ 1~~v~~ v^2]^\rmT$ captures the velocity-dependent components of the uncertainty; it relates to the general safety regressor via $\psi(x) = -\frac{1.8}{\overline{m}}\zeta(v)^\rmT$. The penalty weight is $\rho = 10^3$, and the class $\mathcal{K}$ function is the identity function $\alpha(h(x))=h(x)$.
		We set the bounds and gains as $\theta_{\max} = 20$, $\hat{\theta}(0) = 0$, $\Gamma ={\rm diag}(10^4,10^3,10^2)$, $\gamma = 10^2$, and $\kappa = 1$.
		For this system, $\nabla V(x)=[2(v-v_\rmd)~~0]^\rmT$ and $\phi(x)=L_F V(x) = -\frac{2(v-v_\rmd)}{\overline{m}}[1~~v~~v^2]$, yielding $\|\phi(x)\| = \frac{2\sqrt{V(x)}}{\overline{m}}\sqrt{1+v^2+v^4}$. Since the ego velocity is confined to a bounded operating range $v\in[0,v_{\max}]$ by the bounded wheel force~\eqref{eq:input_constraints} and the dissipative drag, the regressor norm is uniformly bounded on this range. On this compact region, $\sqrt{1+v^2+v^4}\le c_v \isdef \sqrt{1+v_{\max}^2+v_{\max}^4}$, so condition~\ref{eq_111} of Lemma~\ref{prop:asymptotic_dominance} holds with $c_F=2c_v/\overline{m}$, confirming \ref{assum:feasibility} on the operating region.
		
		For comparison, we implement a standard R-CBF. We use a fixed parameter estimate $\theta_\rme = \hat{\theta}(0)=0$. To guarantee safety against all $\theta \in \Theta$, the robust controller subtracts a margin proportional to the worst-case uncertainty. Substituting \eqref{eq:acc_dynamics} into the robust condition \eqref{eq:robust_cbf} generates the constraint
		\begin{equation}
			v_\ell - v - \frac{1.8}{\overline{m}} \left(\|\zeta(v)\| \theta_{\max}+u\right) \geq 1.8 v-d,
		\end{equation}
		which replaces \eqref{eq_safeQP} in the QP formulation.
		
		Figure \ref{fig:acc_results} shows the cost of this robustness. The R-CBF exhibits conservatism, applying mechanical braking far earlier than necessary.
		In contrast, the adaptive controller learns the resistance profile, allowing it to account for drag forces within the safety constraints. This reduces the mechanical braking effort and allows the vehicle to close the gap to the lead vehicle more aggressively than the robust trajectory, while maintaining strict safety.
		
		We quantify performance using three metrics: the minimum safety margin $h_{\min} \isdef \min_{t \in [0,T]} h(x(t))$, which tracks the closest approach to the boundary and where $T=100$~s is the total simulation duration; the total braking effort $E_{\rm brake} \isdef \int_{0}^{T} \big|\min(0, u(t))\big|\,\rmd t$, which measures the mechanical energy spent on deceleration; and the time-averaged clearance $\eta \triangleq \frac{1}{T} \int_{0}^{T} h(x(t)) \, \rmd t$. A lower $\eta$ indicates that the controller operates closer to the constraint, reflecting reduced conservatism.
		The simulation results highlight the efficiency of the adaptive approach. While both controllers require comparable mechanical braking effort $E_{\rm brake} \approx 1.40 \times 10^4$ N$\cdot$s, the robust controller achieves this by maintaining a large, conservative buffer of $h_{\min} = 2.81$ m. In contrast, the CaCBF safely utilizes the available space, approaching within $0.17$ m of the boundary.
		This reduction in conservatism is further confirmed by the average clearance metric: the robust controller operates with a wide margin (i.e., $\eta = 6.27$ m), whereas the adaptive formulation reduces this to $\eta = 0.53$ m, demonstrating that the system recovers nearly all the safe operating space that the robust controller discards.
		The bottom subplot of Figure~\ref{fig:acc_results} confirms that $\hat\theta$ does not converge to $\theta_*$: the ACC velocity profile lacks the persistent excitation required by Definition~\ref{def:pe}. Yet the safety constraint $h(x(t))>0$ is maintained throughout, as guaranteed by Theorem~\ref{thm:main_results0} independently of parameter convergence.
		\exampletriangle
	\end{example}
	
	\begin{figure*}[t!]
		\centering
		\includegraphics[width=.9\textwidth]{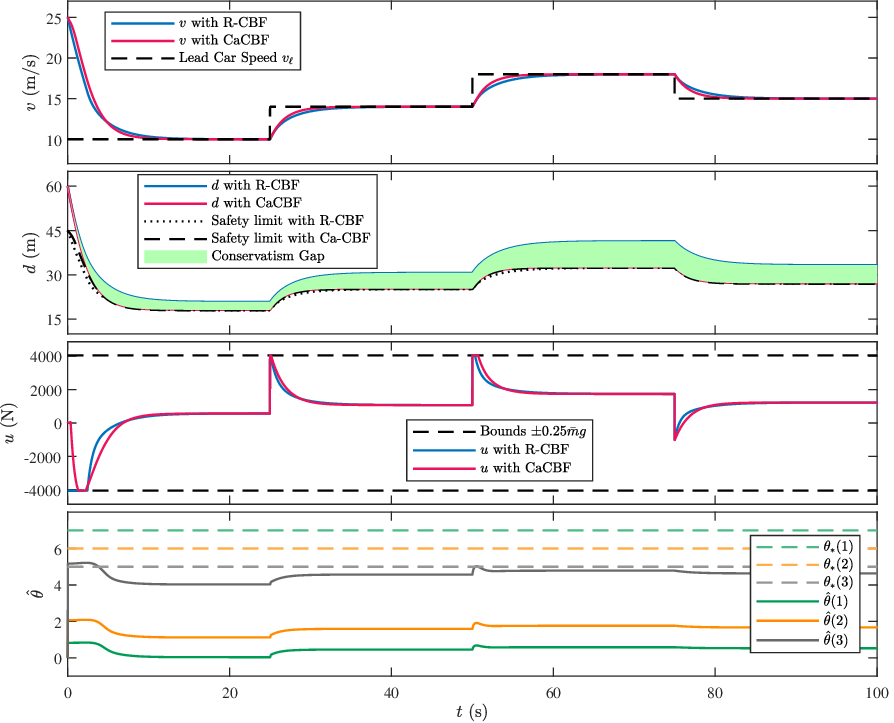} 
		\caption{Example \ref{ex_ACC}: Simulation results for the adaptive cruise control with unknown resistance force. The top subplot shows that the proposed CaCBF (red) tracks the desired speed more effectively than the R-CBF (blue), which exhibits spurious braking.
			The second subplot tracks the inter-vehicle distance $d$, where the green shaded region visualizes the \textit{conservatism gap}: the robust controller maintains an unnecessary buffer due to worst-case assumptions, whereas the adaptive controller operates closer to the effective safety limit $1.8v$. The third subplot plots the control effort $u$, revealing that while both controllers strike the actuation boundaries (dashed black lines), they do so for distinct reasons. The R-CBF is forced into saturation by its conservative worst-case margins, whereas the CaCBF leverages braking authority to maximize tracking performance while remaining strictly safe. The bottom subplot confirms that the parameter estimates $\hat{\theta}$ (solid) need not converge perfectly to the true values $\theta_*$ (dashed) to guarantee safety.}
		\label{fig:acc_results}
	\end{figure*}
	
	\begin{example}\label{ex_Omni}
		\textit{Omnidirectional robot navigation with unknown drift.}
		We illustrate the CaCBF framework using a planar navigation problem for an omnidirectional robot modeled as a kinematic point mass. 
		Let the state be $x =  \matl p_x & p_y \matr^\rmT \in \BBR^2$, where $p_x\in\BBR$ and $p_y\in\BBR$ are the horizontal and vertical positions, respectively. The system evolves according to
		\begin{equation}
			\dot{x} = u + \theta_*,
			\label{eq:omni_dynamics}
		\end{equation}
		where $u = \matl v_x & v_y \matr^\rmT$ is the velocity control input, and $v_x\in\BBR$ and $v_y\in\BBR$ are the horizontal  and vertical velocity, respectively. In addition,  $\theta_* = \matl 0.3 & -0.2 \matr^\rmT$ represents the unknown wind or drift velocity.
		The initial condition is $x(0) = \matl0&5\matr^\rmT$.
		The objective is to navigate from $x(0)$ to a target $x_\rmd \isdef \matl 10 & 5 \matr^\rmT$ while avoiding a  circular obstacle centered at $x_{\rm obs} \isdef \matl 5 & 5 \matr^\rmT$ with radius $r_{\rm obs} \isdef 4.5$~m. We define the safety set $\mathcal{C}$ via the distance-based barrier function
		\begin{equation}
			h(x) = \|x - x_{\rm obs}\|^2 - r_{\rm obs}^2.
		\end{equation}
		To respect actuator limits, we bound the control input such that, for all $t\ge0$,
		\begin{equation}
			\|u(t)\| \leq 2.0 \text{ m/s}.
			\label{eq:input_constraints_omni}
		\end{equation}
		
		We first implement the proposed CaCBF. We select the CLF candidate $V(x) = \|x - x_\rmd\|^2$ with a nominal convergence rate $\lambda = 1.0$.
		The Lie derivatives are $L_f h (x) = 0$ and $L_G h (x)= 2(x - x_{\rm obs})^\rmT$, and the safety regressor is $\psi(x) = 2(x - x_{\rm obs})^\rmT$. Inserting these into the CLF-CBF-QP \eqref{eq:qp_cost_unified}--\eqref{eq:const_clf} yields the  optimization
		\begin{align}
			\big({u}_*(x,\hat\theta), \delta_*(x,\hat\theta)\big) \nn\\
			&\hspace{-3.2em}= \operatorname*{argmin}_{(u, \delta)\in\BBR^2\times[0,\infty)} \ \tfrac{1}{2} \|u - u_{\rm n}\|^2 + \rho \delta^2, \\
			\text{s.t.} \quad & 2(x - x_{\rm obs})^\rmT (u + \hat{\theta}) \geq -\alpha(h(x)), \label{eq_safeQP_omni} \\
			& 2(x - x_\rmd)^\rmT (u + \hat{\theta}) \leq -\lambda V(x) + \delta,
		\end{align}
		where $u_{\rm n}$ is a nominal proportional controller and the penalty weight is $\rho = 10^3$.
		We set the bounds and gains as $\theta_{\max} = 1.0$, $\hat{\theta}(0) = 0$, $\Gamma = 20 I_2$, $\gamma = 10$, and $\kappa = 1$.
		For this system, $\nabla V(x)=2(x-x_\rmd)$ and, since $F(x)=I_2$, the stability regressor is $\phi(x)=L_F V(x)=2(x-x_\rmd)^\rmT$. Hence $\|\phi(x)\|=2\sqrt{V(x)}$, so condition~\ref{eq_111} of Lemma~\ref{prop:asymptotic_dominance} holds with $c_F=2$, confirming \ref{assum:feasibility}.
		
		For comparison, we implement a standard R-CBF. We use a fixed parameter estimate $\theta_\rme = 0$. To guarantee safety against all $\theta \in \Theta$, the robust controller subtracts a margin proportional to the worst-case uncertainty. Substituting the specific geometry of the obstacle into \eqref{eq:robust_cbf} generates the constraint
		\begin{equation}
			2(x - x_{\rm obs})^\rmT u - 2\|x - x_{\rm obs}\| \theta_{\max} \geq -\alpha(h(x)),
		\end{equation}
		which replaces \eqref{eq_safeQP_omni} in the QP formulation.
		
		Figure \ref{fig:omni_results} compares the performance of both controllers. While both methods successfully reach the target without collision, the R-CBF is forced to take a wider, inefficient path. The robust controller maintains an excessive safety buffer due to its worst-case assumption that the wind might always push the robot toward the obstacle.
		In contrast, the adaptive controller learns the drift profile online, allowing it to operate safely with a minimal buffer. By identifying the true wind vector, the CaCBF tracks the obstacle boundary closely, reducing the deviation from the nominal straight-line path and demonstrating the efficiency gains of the proposed framework. The bottom subplot of Figure~\ref{fig:omni_results} shows $\hat\theta$ converging toward $\theta_*$: the curved avoidance trajectory excites both drift components, satisfying the PE condition of Definition~\ref{def:pe}. This convergence, guaranteed by Proposition~\ref{prop:pe_convergence} under PE, drives the tight boundary-following visible in the top subplot. Safety is maintained throughout regardless, by Theorem~\ref{thm:main_results0}.
		
		We quantify performance using four metrics: the minimum safety margin $h_{\min} \isdef \min_{t} h(x(t))$; the total path length $\ell_{\rm path} \isdef \int_{0}^{T} \|u(t)\|\,\rmd t$, where $T=100$~s is the total simulation duration; the time to reach the target $T_{\rm reach} \isdef \min \{t \ge 0 : \|x(t) - x_\rmd\| \le 0.5\}$; and the time-averaged clearance $\eta \isdef \frac{1}{T} \int_{0}^{T} h(x(t)) \, \rmd t$.
		The metrics highlight the operational advantage of the CaCBF approach. While both controllers complete the mission, the R-CBF maintains a larger minimum margin of $h_{\min} = 0.42$ m, which forces it onto a longer path $\ell_{\rm path} \approx 47.72$ m and results in a later arrival time of $T_{\rm reach} = 8.40$ s.
		In contrast, the CaCBF safely utilizes the available space, operating with a much smaller margin of $h_{\min} = 0.04$ m, which allows for a tighter trajectory $\ell_{\rm path} \approx 46.97$ m and faster target convergence in $T_{\rm reach} = 7.42$ s.
		This reduction in conservatism is further confirmed by the average clearance metric: R-CBF operates with a wide margin $\eta = 0.79$ m, whereas CaCBF reduces conservatism to $\eta = 0.52$ m, demonstrating that the system recovers the safe operating space that the R-CBF discards.
		\exampletriangle
	\end{example}
	
	\begin{figure*}[t!]
		\centering
		\includegraphics[width=.9\textwidth]{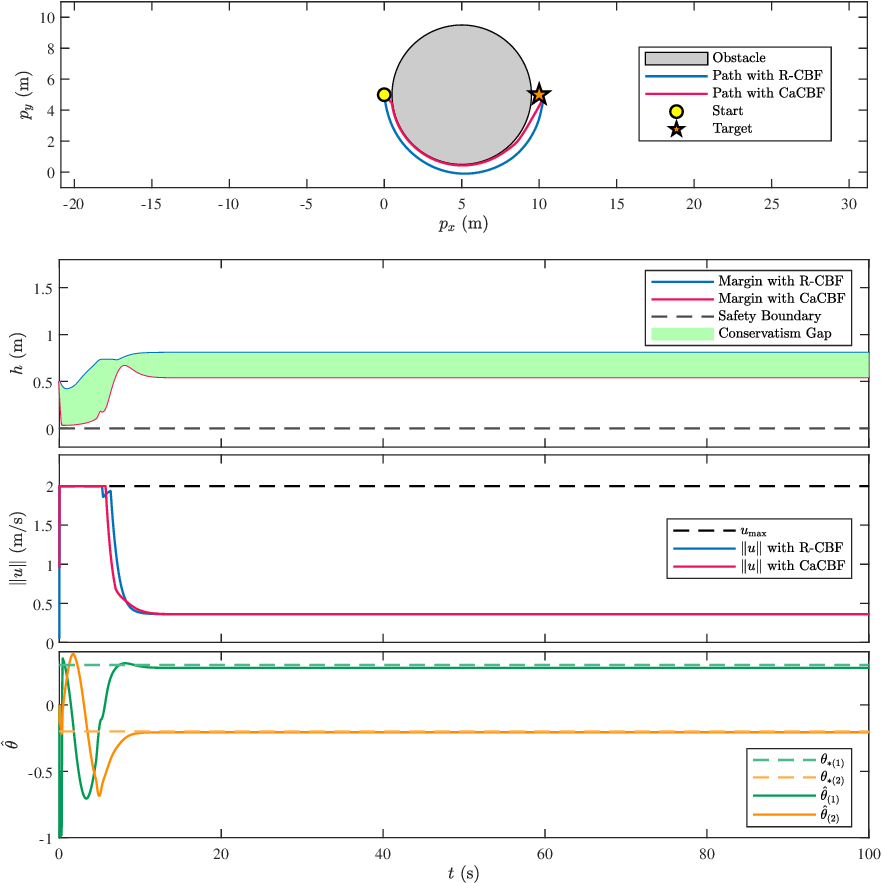} 
		\caption{Example \ref{ex_Omni}: Simulation results for the omnidirectional robot navigation with unknown drift. The top subplot shows the 2D trajectories: while both controllers reach the target (star), the R-CBF (blue) takes a longer route to maintain a robust margin against worst-case drift, whereas the CaCBF (red) adapts to the true drift and hugs the obstacle boundary efficiently. The second subplot tracks the safety margin $h(x)$, where the green shaded region visualizes the \textit{conservatism gap}: R-CBF maintains a larger buffer $h_{\min}=0.42$ m than the adaptive controller buffer $h_{\min}=0.04$ m. The third subplot confirms that the control inputs remain within the saturation limit $\|u\| \le 2.0$~m/s. The bottom subplot shows the parameter estimates $\hat{\theta}$ (solid) converging toward the true drift values $\theta_*$ (dashed), which enables the CaCBF to safely reduce the conservatism gap.}
		\label{fig:omni_results}
	\end{figure*}
	
	\begin{example}\label{ex_Drone_Gate}
		\textit{Planar drone subject to unknown crosswind navigating a passage.}
		We illustrate the capability of the framework to handle complex safety landscapes using a passage scenario. A planar drone modeled as a double integrator must navigate through a tight gate formed by two obstacles while subject to a strong, unknown crosswind. Let the state be $x = \matl p^\rmT & v^\rmT \matr^\rmT \in \BBR^4$, where $p\in\BBR^2$ is position and $v\in\BBR^2$ is velocity. The dynamics are given by
		\begin{align}
			\dot{x} &= \begin{bmatrix} v \\ 0 \end{bmatrix} + \begin{bmatrix} 0 \\ \frac{1}{\overline{m}} I_2 \end{bmatrix} \theta_* + \begin{bmatrix} 0 \\ \frac{1}{\overline{m}} I_2 \end{bmatrix} u,
			\label{eq:drone_dynamics}
		\end{align}
		where $\overline{m}=1$~kg is the mass, $u=\matl u_1&u_2\matr ^\rmT\in\BBR^2$ is the thrust control input, and $\theta_* = \matl 1 & 1.5 \matr^\rmT$ represents the unknown constant wind force.
		The objective is to travel from $p(0) \isdef \matl -8 & 6 \matr^\rmT$ to $p_{\rm d} \isdef \matl 8 & -6 \matr^\rmT$.
		The initial velocity is $v(0)=\matl 0 & 0 \matr^\rmT$. 
		The gate is defined by two circular obstacles centered at $c_1 \isdef \matl -5 & 0 \matr^\rmT$ and $c_2 \isdef \matl 5 & 0 \matr^\rmT$, both with radius $r \isdef 4.8$~m. This creates a physical gap of only $0.4$~m at the narrowest point, which is the origin. The safety set $\mathcal{C}$ is the intersection of the safe regions for both obstacles, defined by the functions
		\begin{equation}
			\bar h_{i}(x) = \|p - c_i\|^2 - r^2, \quad i \in \{1, 2\}.
		\end{equation}
		
		We bound the control input magnitude such that, for all $t\ge0$,
		\begin{equation}
			\|u_{i}(t)\| \leq 10 \text{ N}, \text{~for~} i=1,2.
			\label{eq:drone_input_constraints}
		\end{equation}
		
		Since the position constraints $\bar h_{i}$ have relative degree 2 with respect to the input, we employ the backstepping approach \cite{nguyen2016exponential, xiao2021high}. We define the extended barrier functions $h_{i}(x) \isdef \dot{\bar h}_{i}(x) + 2 \bar h_{i}(x)$.
		We first implement the proposed CaCBF. We utilize a nominal PD controller $u_{\rm n} \isdef 4 (p_{\rm d} - p) - 4 v$ to drive the drone toward $p_{\rm d}$.
		We select the CLF candidate $V(x) = \frac{1}{2} \|v + p - p_{\rm d}\|^2$ based on the error states to enforce convergence with a nominal rate $\lambda = 1$.
		We explicitly calculate the safety Lie derivatives $L_G h_{i} = \frac{2}{\overline{m}}(p - c_i)^\rmT$ and the regressor $\psi_i(x) = \frac{2}{\overline{m}}(p - c_i)^\rmT$.
		Inserting these into the CLF-CBF-QP \eqref{eq:qp_cost_unified}--\eqref{eq:const_clf} yields the optimization
		\begin{align}
			\big(u_*(x,\hat\theta), \delta_*(x,\hat\theta)\big) = \operatorname*{argmin}_{(u, \delta) \in \BBR^2 \times \BBR} \ \tfrac{1}{2} \|u - u_{\rm n}\|^2 +& \rho \delta^2  \\
			\text{s.t.} \quad  \frac{2}{\overline{m}}(p - c_i)^\rmT (u + \hat{\theta}) +2\|v\|^2  + 14(p - c_i)^\rmT v& \geq - 10 \bar h_i(x) , \nn\\ i \in \{1,2\}, \label{eq_safeQP_drone} \\
			 \frac{1}{\overline{m}}(v + p - p_{\rm d})^\rmT (u + \hat{\theta})  +(v + p - p_{\rm d})^\rmT v &\leq - \lambda V(x) + \delta,
		\end{align}
		where the class $\mathcal{K}$ function is $\alpha(h) = 5h$ and the penalty weight is $\rho = 10^3$.
		
		Note that Theorems~\ref{thm:main_results0}, \ref{thm:main_results}, and \ref{thm:noise_robust} require Assumption \ref{assum:C_nonempty_reldegree_1} (i.e, relative degree 1) to hold for the barrier used in the QP. In this example, the extended barriers $h_i(x)=\dot{\bar h}_i(x)+2\bar h_i(x)$ satisfy $L_{G_{\rm aug}}h_i(x)\ne 0$ for all $x\in\mathcal{C}$, where $G_{\rm aug}$ is the effective input matrix after the backstepping extension; Thus, \ref{assum:C_nonempty_reldegree_1}  holds for $h_i$, not for $\bar h_i$. The theorems apply directly to the extended system with $h_i$ playing the role of the barrier function.
		We set the bounds and gains as $\theta_{\max} = 3$, $\hat{\theta}(0) = 0$, $\Gamma = 10 I_2$, $\gamma=10$, and $\kappa = 1$.
		Here $\theta_* = [1~~1.5]^\rmT$ so $\|\theta_*\| \approx 1.80$; the choice $\theta_{\max}=3$ provides a conservative factor-of-$1.67$ overestimate, consistent with the engineering practice of using a safety margin when exact parameter bounds are unavailable. As discussed in Remark~\ref{rem:A1_justification}, overestimating $\theta_{\max}$ increases the initial conservatism of the robust margin but does not affect the validity of the CaCBF safety proof. Let $z\isdef v+p-p_\rmd$, so that $V(x)=\frac{1}{2}\|z\|^2$ and $\nabla V(x)=[z^\rmT~~z^\rmT]^\rmT$. Note that the stability regressor is $\phi(x)=L_F V(x)=\frac{1}{\overline{m}}(v+p-p_\rmd)^\rmT$. Hence $\|\phi(x)\|=\frac{1}{\overline{m}}\|z\|=\frac{\sqrt{2}}{\overline{m}}\sqrt{V(x)}$, and thus condition~\ref{eq_111} of Lemma~\ref{prop:asymptotic_dominance} holds globally with $c_F=\sqrt{2}/\overline{m}$, confirming \ref{assum:feasibility}.
		
		For comparison, we implement  R-CBF using a fixed estimate $\theta_\rme = 0$. To guarantee safety against all $\theta \in \Theta$, the robust controller subtracts a margin proportional to the worst-case uncertainty. Substituting the specific geometry into the robust condition yields the constraints
		\begin{align}
			&2(p - c_i)^\rmT u - 6\|p - c_i\| +2\|v\|^2 \nn\\
			&\quad + 14(p - c_i)^\rmT v \geq  - 10 \bar h_i(x), \quad i \in \{1,2\},
		\end{align}
		which replace \eqref{eq_safeQP_drone} in the QP formulation.
		Figure \ref{fig:drone_multi_objective} illustrates the critical limitation of the robust approach. The R-CBF assumes the worst-case wind could push the drone into the walls, mathematically inflating the obstacles by a robust margin, which is visualized as dotted lines. In this scenario, these margins overlap, effectively closing the gate. Consequently, R-CBF perceives the passage as impassable and braking to a halt before entering the gap.
		In contrast, the CaCBF learns the true wind profile, allowing it to shrink the safety margins dynamically. By adapting to the true wind forces, CaCBF identifies the passage as safe and successfully navigates through the passage.
		
		We quantify performance using four metrics: the minimum safety margin $h_{\min} \isdef \min_{t\in[0,T]} (\min_{i\in\{1,2\}} \bar h_{i}(x(t)))^{1/2}$; the total control effort $E_{\rm control} \isdef \int_{0}^{T} \|u(t)\|\,\rmd t$; the time to reach the target $T_{\rm reach}$; and the average clearance $\eta \isdef \frac{1}{T} \int_{0}^{T} h_{\min}(t) \,\rmd t$, where $T=12$~s is the simulation duration.
		The results confirm the topological advantage of CaCBF. The R-CBF maintains a larger clearance $h_{\min}=0.22$~m and  $\eta=0.35$~m because the constraints preclude entry into the gap, resulting in a failed mission (i.e., timeout at  $T_{\rm reach} = 12$ s).
		The CaCBF, by contrast, operates with a much tighter margin $h_{\min}=0.04$ m to successfully traverse the bottleneck.
		Consequently, the CaCBF controller utilizes higher control effort $E_{\rm control} \approx 40.11$ N$\cdot$s compared to the R-CBF control effort $E_{\rm control} \approx 29.03$ N$\cdot$s as CaCBF performs the necessary maneuvers to complete the mission that is infeasible for R-CBF.
		\exampletriangle
	\end{example}
	
	\begin{figure*}[t!]
		\centering
		\includegraphics[width=.9\textwidth]{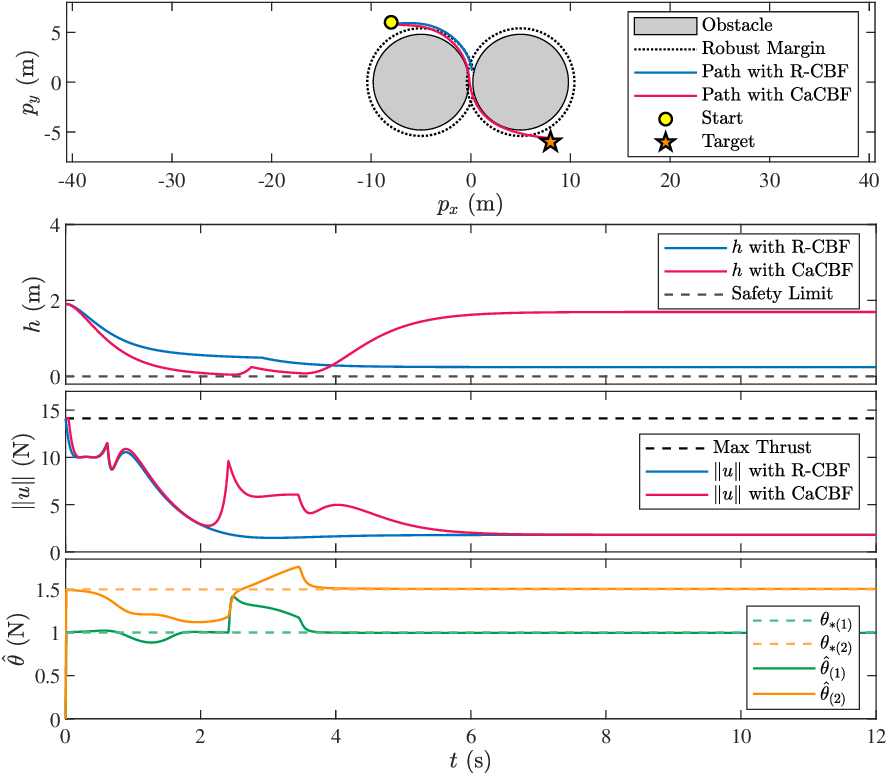} 
		\caption{Example \ref{ex_Drone_Gate}: Simulation results for the planar drone subject to unknown crosswind navigating a  passage. The top subplot shows the trajectories, where the R-CBF (blue) is blocked because the robust margins (dotted) of the two obstacles overlap, effectively closing the gate. The CaCBF (red) learns the wind parameters, shrinks these margins, and successfully passes through the gap to reach the target (star). The second subplot tracks the minimum distance to the nearest wall, confirming that the CaCBF safely utilizes the full available space ($h>0$). The third subplot shows the control effort $\|u\|$, illustrating the maneuvering required to pass the gate. The bottom subplot confirms the convergence of the parameter estimates $\hat{\theta}$ (dotted) toward the true wind values $\theta_*$ (dashed).}
		\label{fig:drone_multi_objective}
	\end{figure*}
	
	\begin{table*}[t]
		\centering
		\small
		\caption{Quantitative summary of the three numerical examples. R-CBF uses a fixed worst-case estimate $\theta_e=0$; CaCBF uses the adaptation law \eqref{eq:adaptation_law}. Lower $\eta$ indicates less conservatism. $^\dagger$ R-CBF fails to reach the target in Example~3 because the inflated robust margins mathematically block the gate.}
		\label{tab:summary}
		\begin{tabular*}{\textwidth}{@{\extracolsep\fill}llcccc@{}}
			\toprule
			\textbf{Example} & \textbf{Method} & $h_{\min}$ (m) & $\eta$ (m) & $E_{\rm brake}$/Control & $T_{\rm reach}$ (s)\\
			\midrule
			\multirow{2}{*}{\ref{ex_ACC} -- ACC} & R-CBF & 2.81 & 6.27 & $1.40\times10^4$ N$\cdot$s & -- \\
			& CaCBF & 0.17 & 0.53 & $1.40\times10^4$ N$\cdot$s & -- \\
			\midrule
			\multirow{2}{*}{\ref{ex_Omni} -- Robot} & R-CBF & 0.42 & 0.79 & 47.72 m & 8.40 \\
			& CaCBF & 0.04 & 0.52 & 46.97 m & 7.42 \\
			\midrule
			\multirow{2}{*}{\ref{ex_Drone_Gate} -- Drone} & R-CBF$^\dagger$ & 0.22 & 0.35 & 29.03 N$\cdot$s & $\infty$ \\
			& CaCBF & 0.04 & -- & 40.11 N$\cdot$s & $<12$ \\
			\bottomrule
		\end{tabular*}
	\end{table*}
	
	\begin{remark}
		\label{rem:kappa_verification}
		We verify that the choice $\kappa=1$ used in all three numerical examples of Section~\ref{sec:simulation} satisfies Corollary~\ref{cor:kappa_explicit}. Since $c_u\ge0$, it follows from \eqref{eq:A_def} that $A\ge 2c_F\theta_{\max}$, and since $\alpha_0\lambda/(2A^2)$ is decreasing in $A$, the bound $\kappa > \frac{\alpha_0\lambda}{2A^2}$ is implied by the stronger condition
		\begin{equation}
			\kappa > \frac{\alpha_0\lambda}{8c_F^2\theta_{\max}^2},
			\label{eq:kappa_sufficient}
		\end{equation}
		which does not require knowledge of $c_u$.
		For Example~\ref{ex_ACC}, $\alpha_0=1$, $\lambda=2$, $\theta_{\max}=20$, and $c_F=2\sqrt{1+v_{\max}^2+v_{\max}^4}/\overline{m}\approx1.09$ for $v_{\max}=30$~m/s and $\overline{m}=1650$~kg. Thus, \eqref{eq:kappa_sufficient} requires $\kappa>5.2\times10^{-4}$.
		For Example~\ref{ex_Omni}, $\alpha_0=1$, $\lambda=1$, $\theta_{\max}=1$, and $c_F=2$. Therefore, \eqref{eq:kappa_sufficient} requires $\kappa>3.1\times10^{-2}$.
		For Example~\ref{ex_Drone_Gate}, $\alpha_0=5$, $\lambda=1$, $\theta_{\max}=3$, and $c_F=\sqrt{2}$. Thus, \eqref{eq:kappa_sufficient} requires $\kappa>3.5\times10^{-2}$.
		In each case $\kappa=1$ satisfies the bound by a factor of at least $29$, confirming that the uniform boundedness guarantee of Theorem~\ref{thm:main_results} applies to all three simulations.
	\end{remark}
	
	\section{Conclusion and Future Work}
	\label{sec:conclusion}
	
	This paper introduced the composite adaptive control barrier function (CaCBF) framework for safety-critical control of nonlinear systems with linear parametric uncertainty. By deriving the adaptation law from a composite energy function that integrates a logarithmic safety barrier, a control Lyapunov function, and a quadratic parameter-error term, CaCBF couples estimation directly to safety. The resulting closed-loop system is formally shown to be: (i) safe, in the sense that the safe set $\mathcal{C}$ is forward invariant for all bounded parameters (Theorem~\ref{thm:main_results0}); (ii) noise-robust, in the sense that the safety guarantee survives bounded state-derivative measurement errors (Theorem~\ref{thm:noise_robust}); and (iii) uniformly ultimately bounded (Theorem~\ref{thm:main_results}). When the regressor is persistently exciting, the parameter estimates additionally converge to the true value exponentially (Proposition~\ref{prop:pe_convergence}). The framework is strictly less conservative than robust CBF in the sense that its admissible control set always contains the robust set as a proper subset (Theorem~\ref{prop:reduced_conservatism}). Simulations on three benchmark systems confirm that CaCBF recovers the safe operating space surrendered by worst-case robust methods while maintaining strict safety throughout.
	
	Several directions merit future investigation. First, the adaptation law requires the state derivative $\dot x$; a natural extension is to replace the instantaneous prediction error with a filtered composite error in the style of Slotine and Li~\cite{slotine1989composite}, which removes the need for direct differentiation while preserving the Lyapunov structure. Second, the linear-in-parameter structure of \eqref{eq:sys_dyn} could be relaxed using universal approximators (neural networks or Gaussian processes) for the unknown term, provided suitable composite energy functions can be constructed. Third, extending CaCBF to high-order relative-degree constraints via the high-order CBF framework of~\cite{xiao2021high} would eliminate the backstepping step used in Example~3. Fourth, the current analysis assumes a fixed safe set $\mathcal{C}$; adaptive or time-varying safe sets arise naturally in multi-agent settings and moving-obstacle scenarios. Finally, hardware validation on physical platforms---where noise, latency, and actuator saturation are all present simultaneously---is essential for assessing deployment readiness.
	
	\bibliographystyle{unsrt}
	\bibliography{ref}
	
\end{document}